\documentclass[reqno,a4paper]{article}

\usepackage[top=1in, bottom=1.25in, left=1in, right=1in]{geometry}
\usepackage[ngerman,british]{babel}

\usepackage[utf8]{inputenc}

\usepackage[T1]{fontenc}
\usepackage{lmodern}

\usepackage{amsmath,amssymb,amsthm,amsfonts}
\usepackage{eurosym}
\usepackage{mathtools}
\mathtoolsset{showonlyrefs}
\usepackage[monochrome]{color}
\usepackage{nicefrac}       

\usepackage{xspace}

\usepackage{float}
\restylefloat{table}
\usepackage{booktabs}

\usepackage{graphicx}
\usepackage{color}
\usepackage{epstopdf}
\DeclareGraphicsExtensions{.pdf,.eps,.png}
\usepackage{subfig}
\usepackage{stackengine}
\usepackage{tikz}
\usetikzlibrary{shapes.misc}

\usepackage{paralist}
\usepackage{enumitem}

\usepackage[short,nocomma]{optidef}

\usepackage{bbm}

\usepackage{calc}

\usepackage{accents}

\newcommand{\Z}{\mathbbm{Z}}

\newcommand{\LL}{\mathcal{L}}

\newcommand{\set}[1]{\{#1\}}

\newcommand{\card}[1]{\lvert #1 \rvert}

\newcommand{\qm}[1]{``#1''}

\newcommand{\ie}{i.e.\ }

\newcommand{\cf}{cf.\ }
\newcommand{\eg}{e.g.\ }

\newcommand{\st}{\mathrm{s.t.}}

\newlength\mysinglespace
\newlength\objspace
\newlength\conspace
\newlength\cconspace
\newenvironment{constr}[1]{\begin{array}[t]{#1}}{\end{array}} 

\newenvironment{opt*}[3]{\begin{equation*}\begin{array}{rl}#1 & #2 \\[\objspace] \st & \begin{constr}{#3}}{\end{constr}\end{array}\end{equation*}\\[0pt]}

\newenvironment{eqns*}[1]{\begin{equation*}\begin{array}[t]{#1}}{\end{array}\end{equation*}\\[0pt]}

\newtheorem{theorem}{Theorem}[section]
\newtheorem{corollary}[theorem]{Corollary}
\newtheorem{definition}[theorem]{Definition}

\newtheorem{lemma}[theorem]{Lemma}

\newcommand{\TT}{\mathcal{T}}
\newcommand{\mS}{\mathcal{S}}

\usepackage[bookmarksnumbered,bookmarksopen,unicode]{hyperref}

\newcommand{\graphicpath}{fig}

\usepackage[]{todonotes}

\DeclarePairedDelimiter{\ceil}{\lceil}{\rceil}
\DeclarePairedDelimiter{\floor}{\lfloor}{\rfloor}

\title{Improving Quantum Computation\\by Optimized Qubit Routing}

\author{
	Friedrich Wagner
	\thanks{Fraunhofer Institute for Integrated Circuits IIS, Erlangen, Germany, \texttt{friedrich.wagner@iis.fraunhofer.de}}
	\and
	Andreas Bärmann
	\thanks{Department of Data Science, University of Erlangen-Nuremberg, Germany, \texttt{andreas.baermann@fau.de}}
	\and
	Frauke Liers
	\thanks{Department of Data Science, University of Erlangen-Nuremberg, Germany, \texttt{frauke.liers@fau.de}}
	\and
	Markus Weissenbäck
	\thanks{Fraunhofer Institute for Integrated Circuits IIS, Erlangen, Germany, \texttt{markus.weissenbaeck@iis.fraunhofer.de}}
}

\begin{document}
	
	\maketitle
\begin{abstract}
In this work we propose a high-quality decomposition approach
for qubit routing by swap insertion.
This optimization problem arises in the context
of compiling quantum algorithms
formulated in the circuit model of computation
onto specific quantum hardware.
Our approach decomposes the routing problem
into an allocation subproblem and a set of token swapping problems.
This allows us to tackle the allocation part
and the token swapping part separately.
Extracting the allocation part
from the qubit routing model of Nannicini et al.\ (\cite{nannicini2021}),
we formulate the allocation subproblem as a binary linear program.
Herein, we employ a cost function
that is a lower bound on the overall routing problem objective.
We strengthen the linear relaxation by novel valid inequalities.
For the token swapping part
we develop an exact branch-and-bound algorithm.
In this context, we improve upon known lower bounds
on the token swapping problem.
Furthermore, we enhance an existing approximation algorithm
which runs much faster than the exact approach
and typically is able to determine solutions close to the optimum.
We present numerical results
for the fully integrated allocation and token swapping problem.
Obtained solutions may not be globally optimal
due to the decomposition and the usage of an approximation algorithm.
However, the solutions are obtained fast
and are typically close to optimal.
In addition, there is a significant reduction
in the number of artificial gates and output circuit depth
when compared to various state-of-the-art heuristics. 
Reducing these figures is crucial for minimizing noise
when running quantum algorithms on near-term hardware.
As a consequence, using the novel decomposition approach
leads to compiled algorithms with improved quality. 
Indeed, when compiled with the novel routing procedure
and executed on real hardware, 
our experimental results for
quantum approximate optimization algorithms
show an significant increase in solution quality
in comparison to standard routing methods.\\
\textbf{Keywords:} Combinatorial Optimization, Integer Programming, Approximation Algorithms, Quantum Compilation\\
\end{abstract}
	\section{Introduction}

\emph{Qubit routing by swap insertion} is a subroutine
in the process of compiling quantum algorithms onto specific hardware.
Quantum algorithms are usually formulated
in the circuit model of quantum computation,
see \eg \cite{NielsenChuang2010} for an introduction.
Such a circuit consists of wires representing qubits
as well as gates representing operations applied to subsets of qubits.
We refer to the qubits in the circuit as \emph{logical} qubits,
in contrast to the \emph{physical} qubits in quantum hardware.
Furthermore, we assume that the circuit
only contains gates acting on at most two qubits,
\ie single and two-qubit gates
(also known as \qm{$ \mathrm{SU}(4) $ circuits}).
This might already be the case for the input circuit to be compiled;
otherwise this can be achieved using the Solovay-Kitaev Theorem,
see \eg \cite[pp.\ 188--202]{NielsenChuang2010}.
In the circuit model, gates may be applied on any pair of logical qubits.
However, this is not the case in many
currently available quantum processors.
Here, two-qubit gates can only be applied on specific pairs
of physical qubits defined by the \emph{hardware connectivity graph}.
Now the task is to choose an initial allocation
of logical qubits in the quantum circuit
to physical qubits in the hardware graph
such that for each two-qubit gate the corresponding logical qubits
are located at neighbouring physical qubits.
The example in Fig.~\ref{fig:example} shows
that this is not always possible directly.

\begin{figure}[h]
	\centering
	\subfloat[]{
		\includegraphics[height=3cm]{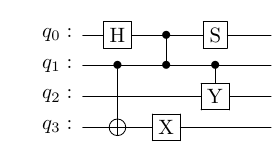}
		\label{fig:example_qc}
	}
	\qquad\qquad
	\subfloat[]{
		\includegraphics[height=3cm]{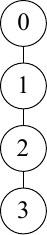}
		\label{fig:example_hg}
	}
	\caption{Example for a quantum circuit (a)
		consisting of four logical qubits
		as well as single and two-qubit gates.
		In the context of this work,
		the meaning of the gates
		is irrelevant. Only the subset of 
		qubits involved in a gate matters.
		The circuit is not compatible
		with the hardware connectivity graph shown in (b).}
	\label{fig:example}
\end{figure}

In such a case, it is necessary
to insert additional \emph{swap gates} into the algorithm.
They effectively swap the positions
of two logical qubits in the hardware graph.
The result is an $ \mathrm{SU}(4) $ circuit
meeting the connectivity restrictions
which is equivalent to the original one
when taking into account the permutation
resulting from initial allocation and swap gates.
A simple example for this procedure
is shown in Figure~\ref{fig:routing_example}.

\begin{figure}
	\centering
	\includegraphics[height=3cm]{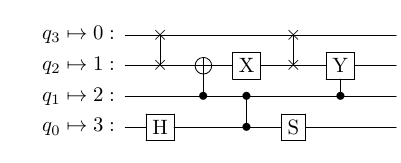}
	\caption{Example for routing by swap insertion.
		The given circuit is equivalent
		to the quantum circuit from Figure~\ref{fig:example_qc}
		and fulfills the hardware connectivity from Figure~\ref{fig:example_hg}.
		The equivalence holds under the initial allocation
		of logical qubits $ \set{q_0, q_1, q_2, q_3} $
		to physical qubits $ \set{0, 1, 2, 3} $ given by
		$ q_0 \mapsto 3, q_1 \mapsto 2, q_2 \mapsto 1, q_3 \mapsto 0 $.
		After the first swap the allocation changes to
		$ q_0 \mapsto 3, q_1 \mapsto 2, q_2 \mapsto 0, q_3 \mapsto 1 $,
		and after the second swap to
		$ q_0 \mapsto 3, q_1 \mapsto 2, q_2 \mapsto 1, q_3 \mapsto 0 $.}
	\label{fig:routing_example}
\end{figure}
When solving the routing by swap insertion problem,
two possible objective functions arise naturally:
either one aims to minimize the execution time on the quantum processor,
which can be achieved by minimizing the output circuit depth,
or one can minimize the total gate number to be executed,
which amounts to minimizing the number of swap gates added.
Intuitively, both objectives are correlated.
However, it is easy to construct examples
where optimal-depth solutions are suboptimal
in terms of gate count and vice versa.
For an experimental study of the influence of depth and gate count
on the performance of quantum algorithms,
we refer to \cite{nannicini2021} and \cite{Scott2021noise}.

\paragraph{Existing methods for qubit routing.}

Many compilation procedures based on routing by swap insertion
have been developed so far.
Some of them may be classified by their local working principle:
they define an initial mapping of logical to physical qubits,
iterate through the circuit in temporal order
and change the allocation by adding swaps
when gates are encountered
which cannot be applied in the current allocation.
Of course, this general principle allows
for many different sophisticated procedures,
in particular when it comes to choosing the initial allocation.
Cowtan et al.'s \emph{Tket compiler}
(see \cite{cowtan2019,Sivarajah2020})
belongs to this class.
It uses a heuristic cost function to choose the next allocation.
Zuhlener et al.\ (see \cite{zulehner2019}) employ an A*-search algorithm
to define the next allocation.
Li et al.'s \emph{Sabre compiler} (see \cite{Li2019})
uses the reversibility of quantum circuits
to perform a bidirectional search for a good initial mapping.
The default compiler in IBM's SDK \emph{Qiskit}, \emph{Stochastic swap},
based on work by Bravi (see \cite{Qiskit}),
as well as the default compiler in Googles SDK \emph{Circ},
\emph{Greedy router} (see \cite{cirq_developers_2021_5182845})
also belong to this class.

A connection from routing by swap insertion
to the problems of token swapping 
and subgraph isomorphism has already
been established and employed
by Siraichi et al.\ in \cite{siraichi2018,Siraichi2019}
and by Childs et al.\ in \cite{childs2019}.
In their work, the subgraph isomorphism problem
arises in the context of searching for high-quality qubit-mappings,
whereas routing between mappings is translated to a token swapping problem.

Two approaches to routing by swap insertion
that do not rely on local swap insertion
but consider the whole circuit instead
are the SAT-based approach by Wille et al.\ (see \cite{wille2019})
and the approach by Nannicini et al.\
based on an integer programming model (see \cite{nannicini2021}).
These approaches suffer from exponential growth of problem size
and NP-hard problem complexity.
In practice, the running times exceed reasonable limits
already for moderate circuit sizes.

Furthermore, there are quantum compilation methods in the literature
which do not rely on swap insertion:
Kissinger and Meijer-Van De Griend consider circuits
only built from CNOT gates
and solve the compiling problem
by finding an equivalent hardware compatible CNOT circuit
(see \cite{Kissinger2020}).
Moro et al.\ use reinforcement learning for quantum compilation
(see \cite{Moro2021}).

Kim (see \cite{Kim2016}) compares
some of the aforementioned routing heuristics
on benchmarks with known optimal solution.
His results reveal a significant margin for improvement.

\paragraph{Our contribution.}

We provide a high-quality solution method
that efficiently exploits the capabilities of integer programming
to solve the routing by swap insertion problem.
The key is to decompose the problem into two separate subproblems,
of which the first one is solved via integer programming.
It determines an allocation sequence
that is compatible with the hardware constraints.
Herein, we employ a cost function
that is a lower bound on the total number of swaps required for routing.
This subproblem is called \emph{token allocation problem}.
For its solution, we develop a simplification of the binary model
introduced by Nannicini et al.\ in \cite{nannicini2021},
which results in a significant reduction in problem size.
Once an optimal allocation sequence has been found,
each pair of subsequent allocations defines a token swapping instance.
These token swapping problems are solved by an approximation algorithm
based on the work of Miltzow et al.\ in \cite{miltzow2016}.
Also, we develop an exact algorithm for token swapping.
By comparison of both solution techniques,
we conclude that the approximation algorithm
typically delivers high-quality solutions
in drastically reduced solution time.

Re-targeting the gates according to the allocation sequence
and inserting the swap gates from the token swapping solutions
finally yields the compiled circuit.
In \cite{Siraichi2019}, a similar decomposition approach, called \emph{Bounded mapping tree} (BMT), is proposed
which, however, is based on dynamic programming.
In contrast to other heuristic methods,
we are able to give bounds on the quality of the obtained solution.
\textcolor{blue}{
We perform extensive numerical experiments on
several benchmark instances from the literature,
covering a broad range of different hardware graphs.
The results show that the proposed approach
finds close-to-optimal solutions
and outperforms well-established heuristics.}
In comparison to the exact approach from \cite{nannicini2021},
on which the present work is based,
our method takes much less time.
This makes the proposed routing procedure applicable
for practically relevant circuit sizes.
Furthermore, our experiments on actual quantum hardware
show that quantum computation benefits from the proposed routing method.

\paragraph{Structure.}

This work is structured as follows.
Section~\ref{sec:prelim} introduces the token allocation problem and the token swapping problem.
\textcolor{blue}{
In Section~\ref{sec:alloc}, we further analyse the token allocation problem
and motivate the use of integer programming for its solution.
We show the NP-hardness of the problem at hand, introduce the binary model and derive valid inequalities for its linear relaxation.
In Section~\ref{sec:ts} we study solution methods for the token swapping problem.
The efficient approximation algorithm employed in our routing procedure
as well as the exact branch and bound algorithm used for benchmarking the former are developed.}
Section~\ref{sec:exp} discusses numerical results for the two subproblems of token allocation
and token swapping as well as for the entire routing problem.
Additionally, we present experimental results for example quantum algorithms executed on actual hardware when routed with different methods.
Finally, in Section \ref{sec:conc} we give a conclusion and indicate directions of further research.

	\section{Preliminaries and Definitions}
\label{sec:prelim}

We start by formally defining the token allocation problem
as well as the token swapping problem.
While the latter is well known,
the former has not been extensively studied yet,
to the authors' best knowledge.
It already
occurred in \cite{nannicini2021} and \cite{siraichi2018,Siraichi2019}
as part of a larger problem.
In the context of qubit routing,
the token allocation problem can be informally described as follows.
Its input is a hardware graph
representing the connectivity of the physical qubits
together with a quantum algorithm acting on a set of logical qubits (\qm{tokens}).
The algorithm is described as a sequence of layers.
In each layer, a set of two-qubit gates
needs to be performed in parallel.
Now the task is to find an allocation for each layer
such that the gates can be executed,
\ie logical qubits involved in gates
are located at neighbouring physical qubits.

We remark that grouping gates into layers
needs to be performed by the routing procedure,
since an algorithm formulated as a circuit 
is not by itself divided into layers
but simply modeled as a sequence of individual gates.
Grouping gates into layers containing more than a single gate
restricts the set of feasible solutions,
since changing allocation between gates in the same layer
is not allowed.
However, grouping significantly reduces the number of layers
and thus the size of the token allocation problem.

For a given allocation sequence, we define its cost
as the sum of all distances logical qubits move
on the hardware graph between subsequent allocations.
Herein, the distance a logical qubit moves between two allocations
is defined as the length of a shortest path connecting the vertices
in the hardware graph to which the qubit is allocated.
The motivation for this choice of objective
is that the number of swaps needed
for routing between subsequent allocations
$ a, a'\colon Q \xrightarrow{1:1} V $
is bounded from below by
\begin{equation}
	\frac{1}{2}\sum_{q \in Q} d_H\left(a(q), a'(q)\right),
\end{equation}
where $ d_H\colon V\times V\rightarrow \Z_0^+ $
denotes the \emph{edge distance} in~$H$,
\ie the number of edges in a path of minimal length
connecting two nodes.
A proof of this result on the token swapping problem
is found \eg in \cite{miltzow2016}. 
This means that instead of minimizing
the total number of swaps required,
we minimize a lower bound on this value
which can be computed much more easily.

We now define the token allocation problem formally
as a combinatorial optimization problem.
\begin{definition}[Token allocation problem]
	An instance of the \emph{token allocation problem (TAP)}
	is given by a tripel $ (H, Q, \Gamma) $, where
	\begin{itemize}
		\item $ H = (V, E) $ is a connected undirected graph,
		\item $Q$ is a set of tokens of size $ \card{Q} = \card{V} $,
		and
		\item $ \Gamma = G^1, G^2, \ldots, G^L $
		is a finite sequence of sets of disjoint token pairs, \ie
		\begin{enumerate}
			\item[(i)] $ G^t = \set{(p^1_t, q^1_t), (p^2_t, q^2_t),
				\ldots, (p^{\card{G^t}}_t, q^{\card{G^t}}_t)}
			\subset Q \times Q,\,
			q^j_t \neq p^j_t\quad
			\forall j \in \set{0, \ldots, \card{G^t}} $ and
			\item[(ii)]	$ \set{p^i_t, q^i_t} \cap \set{p^j_t, q^j_t}
			= \emptyset\quad
			\forall i \neq j \in \set{0, \ldots, \card{G^t}},\,
			\forall t \in \set{1, \ldots, L} $.
		\end{enumerate}
	\end{itemize}
	A feasible solution is given by a sequence $ a_1, \ldots, a_L $
	of bijective mappings $ a_t\colon Q \xrightarrow{1:1} V $
	such that at each time step~$t$ the token pairs in~$ G^t $
	are allocated to neighbouring vertices in $H$:
	\begin{align*}
		\forall 1\leq t \leq L:\, \forall (q_1, q_2) \in G^t:
		\set{a_t(q_1),a_t(q_2)}\in E.
	\end{align*}
	The cost of a feasible solution $ a_1, \ldots, a_L $
	is defined as
	\begin{equation}
		\frac{1}{2} \sum_{t = 0}^{L - 1}
		\sum_{q \in Q} d_H\left(a_t(q), a_{t + 1}(q)\right)\,.
	\end{equation}
	The goal is to find a feasible allocation
	minimizing this cost function.
	\label{def:alloc}
\end{definition}
In the context of qubit routing,
the vertex set~$V$ represents the physical qubits
while the edge set~$E$ represents the set of physical qubit pairs
between which two-qubit gates can be applied.
The logical qubits are represented by the set~$Q$,
and the sets $ G^t \in \Gamma $, $ 1 \leq t \leq L $, 
are the gates to be executed in parallel at time step~$t$.

Next, we give an informal description
of the token swapping problem.
Its input is an undirected, connected graph
and a set of tokens.
Initially, each vertex holds a unique token.
Further, each vertex is assigned a target token.
Now the task is to move tokens along the edges of the graph
such that each vertex holds its desired token.
However, movement of tokens is restricted to swapping adjacent tokens.
Formally, the token swapping problem is defined as follows.
\begin{definition}[Token swapping problem]
	An instance of the \emph{token swapping problem}
	is given by a tupel $ (H, Q, a, a') $, where
	\begin{itemize}
		\item $ H = (V, E) $ is a connected undirected graph,
		\item $Q$ is a set of tokens with size $ \card{Q} = \card{V} $,
		\item $ a\colon Q \xrightarrow{1:1} V $
		is an initial allocation, and
		\item $ a'\colon Q \xrightarrow{1:1} V$ is a final allocation.
	\end{itemize}
	For two nodes $ i, j \in V $, $ i \neq j $,
	the transposition $ (i, j) $
	is the unique bijective mapping from~$V$ onto itself
	interchanging exactly the two elements~$i$ and~$j$.
	Let $ \TT_H $ denote the set of transpositions on~$V$
	restricted to~$E$,
	\ie $ (i, j) \in \TT_H \Leftrightarrow \set{i, j} \in E $.
	A feasible solution $ \mS $
	is a finite sequence of transpositions (\qm{swaps})
	$ \mS = \tau_1, \tau_2 \ldots, \tau_N $,
	where $ \tau_i \in \mathcal{T}_H $ for $ i \in \set{1, \ldots, N} $,
	such that
	\begin{equation}
		a = \tau_N \circ \ldots \circ \tau_2 \circ \tau_1 \circ a'\ .
	\end{equation}
	The cost of a feasible solution is given by its length~$N$.
	The goal of the token swapping problem
	is to find a feasible solution minimizing these costs.
	\label{def:td}
\end{definition}
There is a close link between Problems~\ref{def:alloc} and~\ref{def:td}:
every feasible solution to a TAP instance
defines a sequence of token swapping instances in a canonical way.
In the application to qubit routing,
a solution to this sequence of problems
together with the TAP solution
results in a solution to the qubit routing by swap insertion problem.
Note that this solution might not be optimal in terms of used swaps.
	\section{An Analysis of the Token Allocation Problem}
\label{sec:alloc}

In the following, we analyse the TAP
from Definition~\ref{def:alloc}.
After showing its NP-hardness,
we derive an integer programming formulation
based on edge flows.
Additionally, we strengthen the model
by introducing valid inequalities.

\subsection{NP-hardness}

To motivate the use of integer programming
for the solution of the TAP,
we first show its NP-hardness.
For this purpose, we construct a reduction
of the \emph{subgraph isomorphism problem} to the TAP.
A similar reduction was already sketched by Siraichi et al. in \cite{siraichi2018} to
show NP-completeness of a problem called \emph{qubit assignment problem}.
\textcolor{blue}{This problem is equivalent to the decision version of the TAP.
Here, we transfer their main idea to the TAP}
and work out the reduction in detail.
\begin{definition}[Subgraph isomorphism problem]
	Given two graphs $H$, $H'$,
	the \emph{edge-induced (node-induced)
	subgraph isomorphism problem (SGI)}
	asks whether there is an edge-induced (node-induced) subgraph of~$H$
	isomorphic to~$H'$.
\end{definition}
The complexity of the node-induced SGI is well known.
\begin{theorem}[\cite{pruim2005complexity}]
	The node-induced subgraph isomorphism problem is NP-complete.
\end{theorem}
To transfer the NP-completeness to the edge-induced SGI,
we need the notion of a \emph{line graph}.
\begin{definition}[Line graph]
	The \emph{line graph} $ \LL(G) $ of a graph~$G$
	possesses a node for every edge in~$G$.
	Edges in $ \LL(G) $ connect two nodes
	if and only if the corresponding edges in $G$
	share a common node.
\end{definition}
The proof of the following result is straight-forward and omitted here for brevity.
\begin{lemma}
	Two graphs $H$ and $H'$ are node-induced subgraph isomorphic
	if and only if the corresponding line graphs~$ \LL(H) $
	and~$ \LL(H') $ are edge-induced subgraph isomorphic. 
\end{lemma}
Using this result, the node-induced SGI
can be reduced to the edge-induced SGI.
\begin{corollary}
	The edge-induced subgraph isomorphism problem is NP-complete.
\end{corollary}
For the reduction of edge induced SGI to TAP, we need
\begin{definition}[Connectivity graph]
	\label{def:conng}
	For a finite set of tokens $Q$,
	let $ G \coloneqq
		\set{(q_1, p_1), \ldots, (q_n, p_n)} \subset Q \times Q $,
	with $ q_i\neq p_i $ for $ i = 1, \ldots, n $
	be a set of token pairs.
	Then the \emph{connectivity graph} of $G$
	is given by $ C(G) = (\tilde{Q}, E_C)$ with
	\begin{equation}
		\tilde{Q} \coloneqq \bigcup_{1 \leq k \leq n} \set{q_k} \cup \set{p_k}
		\quad \text{and} \quad
		E_C \coloneqq \set{\set{q, p} \mid (q, p) \in G}.
	\end{equation}
\end{definition}
An example of a connectivity graph
is depicted in Figure~\ref{fig:conng_example}.
The following Lemma establishes a connection between the connectivity graph and the TAP.
The proof is not hard and omitted here. 
\begin{figure}
	\centering
	\includegraphics[height=2cm]{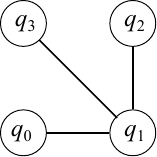}
	\caption{Example for a connectivity graph.
		This example is constructed from the quantum circuit
		in Figure~\ref{fig:example_qc}.
		Logical qubits in the circuit are interpreted as tokens,
		gates are interpreted as token pairs.}
	\label{fig:conng_example}
\end{figure}
\begin{lemma}[\textcolor{blue}{\cite{siraichi2018}}]
	\label{lem:sir}
	A TAP instance $ (H, Q, \Gamma) $ has optimal value of~$0$
	if and only if there is an edge-induced subgraph of $H$
	isomorphic to the connectivity graph \textcolor{blue}{$ C(\bigcup_{G^t\in \Gamma}G^t) $}.
\end{lemma}
With these preliminaries established,
we now study the complexity of the TAP.
\begin{theorem}
	\label{thm:np}
	The token allocation problem is NP-hard.
\end{theorem}
\begin{proof}
    To show NP-hardness,
	we construct a polynomial reduction
	of the edge-induced SGI
	to the decision version of the TAP.
	Let $ H = (V,E) $ and $ H'=(V',E') $ be two graphs.
	For the SGI instance $ (H, H') $,
	we construct an TAP instance $ (H, Q, \Gamma) $.
	Set $ Q \coloneqq V $ and construct~$ \Gamma $ as follows:
	choose an arbitrary order $ 1, \ldots, \card{E'} $
	of the edges in~$ E' $,
	and choose an arbitrary order of the two vertices in each edge.
	This yields a sequence \textcolor{blue}{of edges} $ (p_1, q_1), (p_2, q_2), \ldots,
	(p_{\card{E'}}, q_{\card{E'}}) $, \textcolor{blue}{where}
	$ \set{p_k, q_k} \in E' $
	$ \forall k \in \set{1, \ldots, \card{E'}} $.
	Now, set $ G^k \coloneqq \set{(p_k, q_k)} $
	for all $ k \in \set{1, \ldots, \card{E'}} $ and
	$ \Gamma \coloneqq G^1, \ldots, G^{\card{E'}} $.
	Then we have $ C(\Gamma) = H' $.
	Finally, Lemma~\ref{lem:sir} implies
	that the TAP instance $ (H, Q, \Gamma) $ has optimal value~$0$
	if and only if there is an edge-induced subgraph isomorphism
	from $H'$ to $H$.
\end{proof}
\textcolor{blue}{For NP-hard discrete optimization problems, integer
programming methods such as branch-and-bound and branch-and-cut are
known to be particularly well suited. We will thus follow this route
in the following. 
}
\subsection{Network Flow Model}
\label{sec:edge}

In \cite{nannicini2021}, Nannicini et al.\
formulate the entire routing by swap insertion problem as a network flow problem.
We adjust their model to the TAP.
This results in a model significantly reduced in size
which, of course, is due to the fact
that only a subproblem of routing by swap insertion is modeled.
Let $ H = (V, E) $ be a connected graph,
and let $ (H, Q, \Gamma) $ be a TAP instance.
We introduce the directed arc set
$ A_H \coloneqq \bigcup_{\set{u, v} \in E}
	\set{(u, v)} \cup \set{(v, u)} $
as well as variables with the following interpretations.
The binary variables  $ x^t_{q, i, j} \in \set{0,1} $
take a value of~$1$ if qubit $ q \in Q $
moves from node $ i \in V $ to node $ j \in V $
between time steps $t$ and $ t + 1 $ from $ \set{1, \ldots, L} $,
and~$0$ otherwise.
The auxiliary binary variables $ w_{q, i}^ t \in \set{0,1} $
indicate by a value of $1$ whether qubit $ q \in Q $
is located at node $ i \in V $
in time step $ t \in \set{1, \ldots, L} $,
or~$0$ if not.
Further binary auxiliary variables $ y^t_{(p ,q), (i, j)} \in \set{0,1} $
express whether gate $ (p, q) \in G^t $
is performed along edge $ (i, j) $, value~$1$,
or not, value~$0$.
With these notions, the TAP can be modeled
by the following \textcolor{blue}{quadratic} binary program:
\begin{mini!}
{x}{\displaystyle\sum_{t = 1}^{L-1}\displaystyle\sum_{q\in Q}\displaystyle\sum_{i,j\in V\times V} \textcolor{blue}{d_H(i,j)} x^t_{q,i,j}\label{eq:Objective}}
{\label{eq:edgemodel}}{}
\addConstraint{w^t_{q,i}}			{=\displaystyle\sum_{j\in V}x^t_{q,i,j}\label{eq:flowout} }{\quad\forall 1\leq t < L,\forall i \in V,\forall q \in Q}
\addConstraint{w^t_{q,i}}			{=\displaystyle\sum_{j\in V}x^{t-1}_{q,j,i}\label{eq:fowin} }{\quad\forall 1 < t \leq L,\forall i \in V,\forall q \in Q}
\addConstraint{\displaystyle\sum_{(i,j)\in A_H} y^t_{(p,q),(i,j)}}{=1\label{eq:gates} }{\quad\forall 1\leq t\leq L,\forall (p,q)\in G^t}
\addConstraint{y^t_{(p,q),(i,j)}}{= w_{p,i}^t\cdot w_{q,j}^t\label{eq:ydef}}{\quad\forall 1\leq t\leq L,\forall (p,q)\in G^t,\forall (i,j)\in A_H}
\addConstraint{\displaystyle\sum_{i\in V} w_{q,i}^t}{= 1\label{eq:bij1}}{\quad\forall 1\leq t\leq L,\forall q\in Q}
\addConstraint{\displaystyle\sum_{q\in Q} w_{q,i}^t}{= 1\label{eq:bij2}}{\quad\forall 1\leq t\leq L,\forall i\in V}
\addConstraint{ w_{q,i}^t}{\in \{0,1\}\label{eq:binw}}{\quad\forall 1\leq t\leq L,\forall q\in Q,\forall i\in V}
\addConstraint{ x^t_{q,i,j}}{\in \{0,1\}\label{eq:binx}}{\quad\forall 1\leq t\leq L-1,\forall q\in Q,\forall (i,j)\in V\times V}
\addConstraint{ y^t_{(p,q),(i,j)}}{\in \{0,1\}\label{eq:biny}}{\quad\forall 1\leq t\leq L,\forall (p,q)\in G^t,\forall (i,j)\in A_H.}
\end{mini!}
\\[-\baselineskip]
Constraints~\eqref{eq:flowout} and~\eqref{eq:fowin}
ensure logical qubit conservation.
Via Constraints~\eqref{eq:gates},
we enforce that every gate is implemented.
Constraints~\eqref{eq:ydef} demand that a gate
be implemented along an arc if and only if logical qubits
are located at the physical qubits of the arc.
Finally, Constraints~\eqref{eq:bij1} and~\eqref{eq:bij2}
establish bijective mappings between logical and physical qubits
while Constraints~\eqref{eq:binw},\eqref{eq:binx},\eqref{eq:biny}
define the domains of the variables.

Model~\textcolor{blue}{\eqref{eq:Objective}-\eqref{eq:biny}} can be illustrated
by a time-expanded hardware graph~$G$,
which contains a copy of the nodes in~$H$
for every time step
as well as edges connecting subsequent time steps.
An example is shown in Figure~\ref{fig:TimeExpandedGraph}.
\textcolor{blue}{A feasible solution is represented by a collection of $|Q|$ many vertex-disjoint paths in $G$,
	where each path starts at time step $t=1$ and ends at time step $t=L$.
	Furthermore, in every time step $t$, neighboring restrictions
	implied by the gate group $G^t$ need to be satisfied by the paths.}
\begin{figure}
	\centering
	\includegraphics[width=0.45\linewidth]{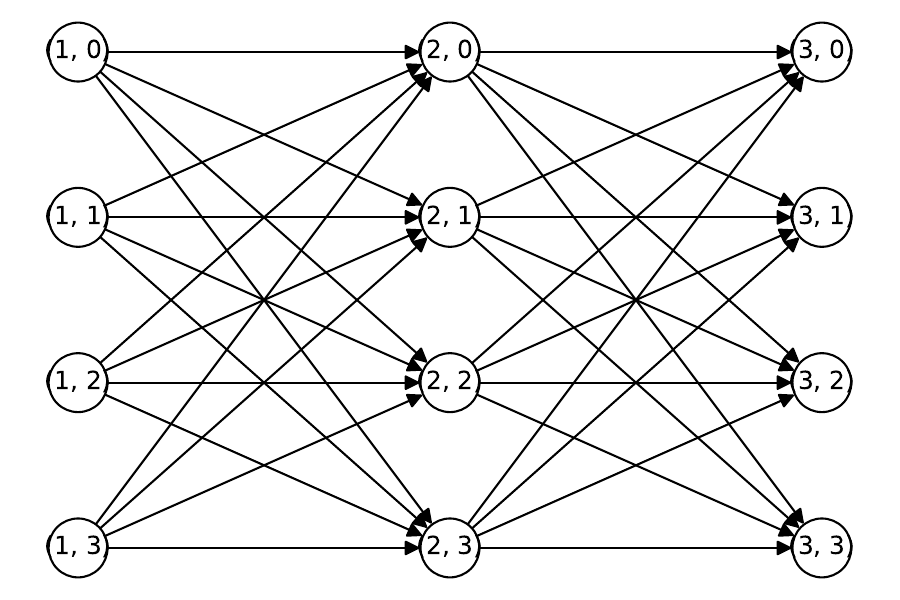}
	\caption{Example for the time-expanded hardware graph~$G$
		on which the network flow formulation of the TAP
		is based.
		The instance is taken
		from the example in Figure~\ref{fig:example}.
		Accordingly, there are $ \card{Q} = 4 $ logical qubits
		and $ L = 3 $ time steps.
		The first number in a node labels the time step
		while the second number labels the physical qubit.}
	\label{fig:TimeExpandedGraph}
\end{figure}

For our further analysis of the model,
we linearize \textcolor{blue}{the quadratic} Constraints~\eqref{eq:ydef}
by replacing them with
\begin{align}
	\label{eq:mccormick}
	y^t_{(p, q), (i, j)} &\leq w_{p, i}^t\notag\\
	y^t_{(p, q), (i, j)} &\leq w_{q, j}^t\\
	y^t_{(p, q), (i, j)} &\geq w_{p, i}^t + w_{q, j}^t - 1.
\end{align}
\textcolor{blue}{From now on, we only consider the linearized version of  \eqref{eq:Objective}-\eqref{eq:biny}
	 and refer to it as Model \eqref{eq:edgemodel}}.
\paragraph{Optimal solution to the linear relaxation.}

To illustrate the need for the introduction
of strengthening inequalities,
we show that the \textcolor{blue}{linear programming (LP)} relaxation
of Model~\eqref{eq:edgemodel} is very weak
in the sense that for \emph{any} TAP instance
there is an optimal LP solution with value~$0$.
To see this, we assign the flow value $ 1 / \card{Q} $
to all paths in~$G$ with zero costs.
More formally, we set
\begin{align*}
	x^t_{q, i, j} &= \begin{cases}
		1/|Q|, & i = j\\
		0, & i \neq j
		\end{cases}
		\qquad \forall i, j \in V,\,
			\forall q \in Q,\,
			\forall 1 \leq t \leq L - 1,\\
	w^t_{q, i} &= 1 / \card{Q}
		\quad \forall i \in V,\,
		\forall q \in Q,\,
		\forall 1 \leq t \leq L\\
	\intertext{and}
	y^t_{(p, q), (i, j)} &= 1 / \card{A_H}
		\quad \forall 1 \leq t \leq L,\,
			\forall (p, q) \in G^t,\,
			\forall (i, j) \in A_H.
\end{align*}
It can be checked that this fully symmetrical, fractional solution
is valid for the LP relaxation of Model~\eqref{eq:edgemodel}.
Concerning Constraints~\eqref{eq:mccormick},
note that $ \card{A_H} \geq 2(\card{Q} - 1) \geq \card{Q} $
and $ 2 / \card{Q} - 1 \leq 0 $ hold,
since~$H$ is connected and we have $ \card{Q} \geq 2 $.

\subsection{Subgraph Isomorphism Inequalities}
\label{sec:subgi}

We will now derive valid inequalities for Model~\eqref{eq:edgemodel}
to strengthen its LP relaxation.
To this end, we extend the result of Lemma~\ref{lem:sir},
which we also used in the proof of Theorem~\ref{thm:np}.
For a given TAP instance $ (H, Q, \Gamma) $,
this will enable us to identify a qubit subset $ \tilde{Q}\subseteq Q $,
an integer distance $ d \geq 0 $ in $H$
as well as time steps~$ t_0 $ and~$ t_1 $,
between which at least one qubit pair in~$\tilde{Q}$
has to move a total distance of at least~$ d + 1 $.
For the derivation we need
\begin{definition}[$d$-th relaxed graph]
	\label{def:hd}
	For an undirected, connected graph $ H = (V, E) $
	and an integer $ d \geq 0 $,
	we define the \emph{$d$-th relaxed graph} $ H^d = (V, E^d) $
	as the graph defined by
	\begin{align*}
		\set{i, j} \in E^d \iff d_H(i, j) \leq d + 1\ . 
	\end{align*}
\end{definition}
Note that $ H^0 = H $.
An example for Definition~\ref{def:hd}
is shown in Figure~\ref{fig:Hd}.
\begin{figure}[]
	\centering
	\subfloat[$ H^0 = H $]{
		\includegraphics[height=1.5cm]{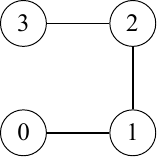}
		\label{fig:example_h0}
	}
	\qquad
	\subfloat[$ H^1 $]{
		\includegraphics[height=1.5cm]{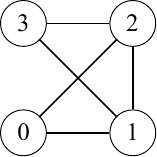}
		\label{fig:example_h1}
	}
	\qquad
	\subfloat[$ H^2 $]{
		\includegraphics[height=1.5cm]{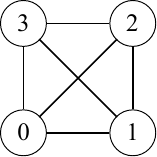}
		\label{fig:example_h2}
	}
	\caption{Example for the $d$-th relaxed graphs $ H^d $
		from Definition~\ref{def:hd}.
		The graph~$H$ is the four-qubit line graph
		from the example in Figure~\ref{fig:example_hg}.}
	\label{fig:Hd}
\end{figure}
We are now able to sketch the main idea
for deriving valid inequalities.
Choose a starting time~$ t_0 $ and an end time~$ t_1 $
with $ 1 \leq t_0 < t_1 \leq L $
as well as a subset~$ \tilde{G} $
of the gates between~$ t_0 $ and~$ t_1 $.
Now, consider the connectivity graph $ C(\tilde{G}) = (\tilde{Q}, E_C) $.
We will relate it to the $d$-th relaxed graph~$ H^d $ of~$H$
for some integer $ d \geq 0 $.
Namely, if there is no isomorphism from $ C(\tilde{G}) $ to some edge-induced subgraph 
of~$ H^d $, we know that there is no allocation
such that all qubits involved in the gates of~$ \tilde{G} $
are at most an edge distance of $ d + 1 $ apart.
Consequently, there has to be at least one gate in $ \tilde{G} $
such that the involved qubits are more than~$ d + 1 $ apart at~$ t_0 $.
Between time steps~$ t_0 $ and~$ t_1 $,
this qubit pair will move a total distance of at least~$ d+1 $.
This leads us to the following
\begin{lemma}[Subgraph isomorphism inequalities]
	\label{lem:sgi}
	Let $ (H, Q, \Gamma) $ be a TAP instance with $ H = (V, E) $
	and $ \Gamma = G^1, \ldots, G^L $.
	Further,
	let \smash{$ \tilde{G} \coloneqq \bigcup_{t = 1}^L \tilde{G}^t $},
	$ \tilde{G}^t \subseteq G^t $,
	be a set of gates
	and let $ C \coloneqq C(\tilde{G}) = (\tilde{Q}, E_C) $
	be the corresponding connectivity graph.
	Finally, let $ d \geq 0 $ be an integer
	and let~$ H^d $ be the $d$-th relaxed graph of~$H$.
	Now consider the time steps
	$ t_0 \coloneqq \min\set{t \mid \tilde{G}^t \neq \emptyset} $
	and $ t_1 \coloneqq \max\set{t \mid \tilde{G}^t \neq \emptyset} $.
	If there is no edge-induced subgraph of~$ H^d $ isomorphic to~$C$,
	then the inequalities
	\begin{align}
		\sum_{t = t_0}^{t_1 - 1} \sum_{q \in \tilde{Q}}
			\sum_{(i, j) \in V \times V} d_H(i, j) x^t_{q, i, j}
				&\geq d + 1
		\label{eq:sgi1}
		\intertext{and}
		\sum_{t = t_0}^{t_1 - 1} \sum_{q \in Q}
			\sum_{(i, j) \in V \times V} d_H(i, j)x^t_{q, i, j}
				&\geq 2(d + 1)
		\label{eq:sgi2}
	\end{align}
	hold.
\end{lemma}
\begin{proof}
	First, we show the validity of~\eqref{eq:sgi1}.
	For a feasible binary solution to Model~\eqref{eq:edgemodel},
	consider the allocation~$ a_{t_0} $ of qubits at time step~$ t_0 $.
	From $ a_{t_0} $, construct the edge-induced subgraph
	$ H' = (V', E') $
	of $ H^d = (V, E^d) $
	defined by the edge subset
	\begin{align}\label{eq:eprime}
		E' \coloneqq \left\set{\set{i, j} \in E^d \mid
			a_{t_0}^{-1}(i), a_{t_0}^{-1}(j) \in \tilde{Q}
				\wedge \set{a_{t_0}^{-1}(i), a_{t_0}^{-1}(j)}
					\in E_C\right}.
	\end{align}
	The mapping defined by the allocation~$ a_{t_0} $
	is not an isomorphism from~$ H' $ to~$C$,
	since~$C$ is not subgraph isomorphic to~$ H^d $ by assumption.
	Furthermore, the definition \eqref{eq:eprime} of~$E'$ implies,
	that for every edge $ \set{i, j} $ in~$ H' $
	there is a corresponding edge
	$ \set{a_{t_0}^{-1}(i), a_{t_0}^{-1}(j)} $ in~$C$.
	This means that there is an edge $ \set{q_1, q_2} $ in~$C$
	such that $ \set{a_{t_0}(q_1), a_{t_0}(q_2)} $
	is not an edge in~$ H' $.
	\textcolor{blue}{Thus, $ \set{a_{t_0}(q_1), a_{t_0}(q_2)} $ is also not an edge in $H^d$} by equation \eqref{eq:eprime}.
	This edge corresponds to a gate $ (q_1, q_2) \in G^{t'} $
	for some~$ t' $ with $ t_0 < t' \leq t_1 $.
	For this gate, the logical qubits~$ q_1 $ and~$ q_2 $
	are located at physical qubits more than~$ d + 1 $ apart at~$ t_0 $,
	\ie
	\begin{equation}
		\label{eq:distq1q2t0}
		d_H(a_{t_0}(q_1), a_{t_0}(q_2)) > d + 1,
	\end{equation}
	but need to be at neighbouring physical qubits at~$ t' $.
	Thus, together they need to move a distance of at least $ d + 1 $:
	\begin{equation}
		\label{eq:dist1q2traveled}
		d_H(a_{t_0}(q_1), a_{t'}(q_1)) + d_H(a_{t_0}(q_2), a_{t'}(q_2)
			\geq d + 1.
	\end{equation}
	This implies~\eqref{eq:sgi1}.
	Now, we show the validity of~\eqref{eq:sgi2}.
	For a subset of vertices $ \tilde{Q} \subseteq Q $ let
	\begin{equation}
		L(\tilde{Q}) \coloneqq \sum_{q \in \tilde{Q}}
			d_H\left(a_{t_0}(q), a_{t'}(q)\right)
	\end{equation}
	\textcolor{blue}{denote the sum of distances between
	qubit allocations at $t_0$ and qubit allocations at $t'$
	for logical qubits in~$ \tilde{Q} $.}
	Then, for the left-hand side of~\eqref{eq:sgi1} it holds
	\begin{equation}
		\sum_{t = t_0}^{t_1 - 1} \sum_{q \in \tilde{Q}}
			\sum_{(i, j) \in V \times V} d_H(i, j) x^t_{q, i, j}
				\geq L(\tilde{Q}),
	\end{equation}
	since $ t' \leq t_1 $.
	Consider the $Q$-permutation
	$ \pi \coloneqq a_{t_0}^{-1} \circ a_{t'}\colon Q \rightarrow Q $.
	Let $ K_1, K_2 \subseteq Q $ be the cycles in~$\pi$
	containing~$ q_1 $ and~$ q_2 $, respectively.
	First, consider the case~$ K_1 \neq K_2 $.
	Here it holds
	\begin{equation}
		L(\tilde{Q}) \geq L(K_1) + L(K_2)\ ,
	\end{equation}
	since $ K_1 \cap K_2 = \emptyset $.
	Furthermore, we know $ L(K_i) \geq 2d_H(a_{t_0}(q_i), a_{t'}(q_i)) $,
	$ i \in \set{1, 2} $,
	since~$ K_1 $ and~$ K_2 $ are cycles
	and $ d_H(\cdot, \cdot) $ satisfies the triangle inequality,
	see Figure~\ref{fig:Keq} for an illustration.
	\begin{figure}[]
		\centering
		\subfloat[$ K_1 \neq K_2 $]{
			\includegraphics[height=2cm]{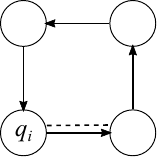}
			\label{fig:Keq}
		}
		\qquad\qquad
		\subfloat[$ K_1 = K_2 $]{
			\includegraphics[height=2cm]{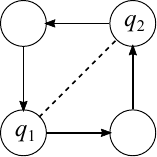}
			\label{fig:Kneq}
		}
		\caption{Sketch for the two cases
			in the proof of Lemma~\ref{lem:sgi}.
			The cycle in~(a) contains
			exactly one of the qubits~$ q_1 $ and~$ q_2 $
			and its length is greater or equal
			than twice the distance this qubit moves (dashed line).
			The cycle in~(b) contains both qubits~$ q_1 $ and~$ q_2 $,
			and its length is greater or equal
			than twice the initial distance of~$ q_1 $ and~$ q_2 $
			(dashed line).}
		\label{fig:K}
	\end{figure}
	Thus, we have
	\begin{equation}
		L(K_1) + L(K_2) \geq 2d_H(a_{t_0}(q_1), a_{t'}(q_1))
			+ 2d_H(a_{t_0}(q_2), a_{t'}(q_2))
			\overset{\eqref{eq:dist1q2traveled}}{\geq} 2(d + 1).
	\end{equation}
	On the other hand, if $ K_1 = K_2 \eqqcolon K $, we find
	\begin{equation}
		L(\tilde{Q}) \geq L(K).
	\end{equation}
	Thus, again using the triangle inequality
	and the fact that~$K$ is a cycle,
	we have
	\begin{equation}
		L(K) \geq 2d_H(a_{t_0}(q_1), a_{t_0}(q_2))
			\overset{\eqref{eq:distq1q2t0}}{\geq} 2(d + 1)
	\end{equation}
	as illustrated in Figure~\ref{fig:Kneq}.
	This implies~\eqref{eq:sgi2}.
\end{proof}
Clearly, any valid subgraph isomorphism inequality
cuts off all solutions of the LP relaxation with zero cost.
However, generating these inequalities can be computationally expensive,
since subgraph isomorphism is NP-complete.

Having shown its NP-hardness and analysed its binary model,
we conclude this section about the TAP
by introducing methods to reduce the model size for larger instances.

\subsection{Algorithmic Enhancements for Large Instances}\label{sec:practical}

To keep the size of Model~\eqref{eq:edgemodel} tractable
for larger hardware graphs and quantum circuits,
we describe two heuristic methods for eliminating variables,
which means that both methods
may in principle lead to suboptimal solutions.
The number of flow variables is decreased
by either removing edges from the time-expanded graph~$G$
or by removing logical qubits not participating in two-qubit gates.

\paragraph{Limiting Distance.}

The first method eliminates edges from~$G$.
We only consider edges which connect physical qubits in~$H$
that are within a given distance limit.
Thus, we limit the distance that a logical qubit may move
between subsequent allocations.
\textcolor{blue}{A method to find a reasonable distance limit is to iteratively increase the limit until the TAP becomes feasible.
Once the model is feasible, the distance is again increased by one and the TAP is solved a last time.}

\paragraph{Limiting the Number of Qubits.}

In practice, the number of logical qubits is usually smaller
than the number of physical qubits.
In principle, this can be reduced to the case
of an equal number of logical and physical qubits
by adding inactive logical qubits not participating
in any two-qubit gates.
However, the number of variables in the TAP model
grows linearly with $ \card{Q} $.
Therefore, we only consider flow variables associated with active logical qubits.
This bears the problem
that costs associated with the flow of inactive qubits
are not counted.
To circumvent this issue, we do not allow active logical qubits
to change positions with inactive qubits.
Thus, the subgraph of~$H$ at which active logical qubits are located
is fixed for all time steps.
\textcolor{blue}{However, we stress that the particular choice of this subgraph remains subject to optimization.}

	\section{Solving the Token Swapping Problems}\label{sec:ts}

We now turn to the solution of the token swapping problems
which need to be solved as the second step
in the proposed routing procedure,
once an optimal TAP solution has been found.
Token swapping is known to be NP-hard (see \eg \cite{miltzow2016}).
We will develop two solution methods
for the token swapping problem:
an efficient approximation algorithm
as well as an exact branch-and-bound approach.

\subsection{Approximate Solution}\label{sec:miltzow}

The proposed qubit routing procedure employs a modified version of the
efficient token swapping algorithm given by Miltzow et al.\ in \cite{miltzow2016}.
\textcolor{blue}{This approximation algorithm has running
  time that is bounded by a polynomial in the number of tokens, also in the worst case. 
}
It has a guaranteed approximation factor of four in terms of swap count.
However, we note that on all instances we considered,
the algorithm exhibits an approximation factor better than 1.5.
Indeed, the analysis in \cite{miltzow2016} shows
that the approximation factor is strictly less than four. 

We briefly review the original approximation algorithm.
For a detailed description and analysis,
we refer the interested reader to \cite{miltzow2016}.
Given a token swapping instance $ (H, Q, a, a') $,
the algorithm iteratively performs swaps
by gradually changing the allocation from $a$ to $ a' $.
We adapt the notion of \cite{miltzow2016}
for an arbitrary token allocation~$ \tilde{a} $:
an \emph{unsatisfied vertex} is a vertex $ i \in H $
which does not hold its target token,
\ie $ \tilde{a}^{-1}(i) \neq (a')^{-1}(i) $.
A \emph{distance-decreasing neighbour} of a vertex $ i \in H $
is a neighbour of $i$
from which the distance to the target vertex
of the token sitting at $i$
is strictly less than the distance from $i$.
A \emph{happy swap chain} of length $k$ is a sequence of $k$ swaps
$ (i_1, i_2), (i_2, i_3), \ldots, (i_{k - 1}, i_k), (i_k, i_{k + 1})
	\subseteq \TT_H $
such that every token involved
has its distance decreased
after the swaps in the chain have been performed sequentially.
An \emph{unhappy swap} is a swap that moves a
token already positioned 
at its target vertex, \textcolor{blue}{while the other token moves closer to its target}.
Given these notions, the algorithm works as follows:
\begin{enumerate}
	\item Choose an unsatisfied vertex.
	\item Perform a walk along distance-decreasing neighbours until a cycle is closed or a dead end encountered, \ie no distance-decreasing neighbour exists.
	\item Perform the happy swap chain or unhappy swap.
	\item Go to 1.
\end{enumerate}
The key insight why this algorithm terminates is,
that vertices involved in an unhappy swap
will be part of a future happy swap chain.
Furthermore, the number of happy swap chains is finite since they necessarily move tokens further towards their targets.
\textcolor{blue}{The original algorithm,
as described in \cite{miltzow2016}, chooses random vertices in steps 1 and 2.}
We modify the algorithm in two ways aiming to further reduce both resulting swap count and depth.
First, in step 1, we choose an unsatisfied vertex
that was not part of the previously performed swap chain or unhappy swap.
With this choice, we try to reduce depth:
only swaps on disjoint vertices can be performed in parallel.
Second, we modify step 2.
\textcolor{blue}{In each step of the walk}
we try to choose cleverly
among all distance-decreasing neighbours
by exploring them first.
If possible, we choose
the one closing the smallest cycle.
The motivation for this is, that all swaps in a
happy swap chain cannot be performed in parallel since they share common vertices.
If no cycle can be closed,
we try to avoid dead ends, \ie unhappy swaps.
Intuitively, a typical optimal
solution does not use many unhappy swaps.

Having introduced an efficient but approximate method, we now derive an exact branch and bound algorithm.

\subsection{Exact Solution}
The exact approach will be used for benchmarking
the approximation algorithm presented in the previous section.
The exact method finds a shortest path from the initial allocation
to the final allocation in the associated \emph{Cayleigh graph}.
The latter possesses a node for each of the $ \card{Q}! $ possible allocations.
Two nodes are adjacent if and only if there is a swap,
\ie a transposition $ \tau \in \mathcal{T}_H $,
which maps between the associated allocations.
It follows that an optimal solution to a token swapping problem
is given by a shortest path from the initial allocation
to the final allocation in the Cayleigh graph.
The search algorithm works through the nodes of the Cayleigh graph starting at the initial allocation.
At each node, we give upper and lower bounds on the shortest path length
to the goal node.
Upper bounding is achieved by running the modified 4-approximation algorithm
from Section~\ref{sec:miltzow},
yielding a path from the current node to the goal node.
This path is augmented by the path from the start node to the current node,
giving a feasible solution to the token swapping problem.

We use two different types of lower bounds, described in the following sections.

\subsubsection{Improved Distance Lower Bound}

A well-known lower bound on token swapping,
already introduced as objective of the TAP in Section \ref{sec:alloc},
is given by
\textcolor{blue}{
\begin{lemma}[\cite{YAMANAKA2015,miltzow2016} Distance lower bound]
	\label{lem:halfdistlb}
	For any token swapping instance {$ T = (H, Q, a, a') $} it holds
	\begin{align}\label{eq:halfdistlb}
		\mathrm{OPT}(T) \geq \ceil*{\frac{1}{2}\sum_{q \in Q} d_H(a(q), a'(q))}.
	\end{align}
\end{lemma}
The following Lemma describes a class of instances for which this is bound sharp,\ie satisfied with equality.
\begin{lemma}[]
	\label{lem:lbsharpness}
	The distance lower bound \eqref{eq:halfdistlb} is sharp on all token swapping instances that require two or less swaps.
\end{lemma}
\begin{proof}
	It is easy to see that the bound is sharp for an instance that requires one or less swaps.
	Now, consider an instance that requires two swaps.
	Here, exactly a two cases can occur. Either the two swaps act on disjoint vertices, or they share a common vertex.
	In the first case, the sum in \eqref{eq:halfdistlb} amounts to 4, where in the second case it amounts to at least 3.
\end{proof}
A simple example for a token swapping instance requiring two swaps is given by a 3-vertex line graph where initially no vertex holds its desired token.
On the contrary, for every optimal objective value greater or equal to 3,
one can construct token swapping instances,
for which the distance bound is not sharp.
On a complete (sub-)graph, place $n$ tokens such that the underlying permutation has exactly one cycle.
Then, an optimal solution has $n-1$ swaps as we will see in the following section.
However, the distance bound is $\ceil*{n/2} < n-1$ for $\geq 3$.
}

\textcolor{blue}{
Lemma~\ref{lem:lbsharpness} implies that the decomposition
of the qubit routing problem into TAP and token swapping is guaranteed
to return optimal solutions
for all instances that have a solution with two or less swaps.
}
Our aim is now to improve upon lower bound~\eqref{eq:halfdistlb}.
To this end, we introduce the notion of a \emph{blocking vertex}.
\begin{definition}[Blocking vertex]
	\label{def:blocking}
	Let $ (H, Q, a, a') $ be a token swapping instance.
	Let $ q \in Q $ be a token not yet sitting at its target vertex,
	\ie $ a(q) \neq a'(q) $.
	A vertex $ v \in H $ is called \emph{$q$-blocking vertex}
	if it is contained in a shortest path in $H$ from $ a(q) $ to $ a'(q) $,
	and if it holds its desired token
	\ie $ a^{-1}(v) = (a')^{-1}(v) $.
\end{definition}
Intuitively, the presence of blocking vertices
increases the number of required swaps.
This intuition is confirmed by
\begin{lemma}[Improved distance lower bound]
	\label{lem:distlb}
	Let \textcolor{blue}{$ T = (H, Q, a, a') $} be a token swapping instance.
	Let $ Q^* \subseteq Q $ denote the set of tokens
	not yet sitting at their target vertices.
	For every $ q \in Q^* $
	let $ B_q \subset V $ be a set of $q$-blocking vertices
	(see Definition~\ref{def:blocking}) of size $ k_q $
	such that
	\begin{enumerate}
		\item[(i)] any path in $H$ from $q$ to its target vertex
			avoiding any of the blocking vertices in $ B_q $
			has at least length $ d_H(a(q), a'(q)) + 2k_q $ and
		\item[(ii)] the sets $ B_q $ are pairwise disjoint.
	\end{enumerate}
	Then it holds
	\begin{align}
		\mathrm{OPT}(\textcolor{blue}{T}) \geq \ceil*{\sum_{q \in Q} \frac{d_H(a(q), a'(q))}{2}}
			+ \sum_{q \in Q^*} k_q.
	\end{align}
\end{lemma}
\begin{figure}
	\centering
	\includegraphics[width=0.15\linewidth]{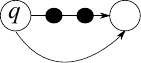}
	\caption{Sketch for the proof of Lemma~\ref{lem:distlb}.
		The shortest path from the vertex holding token $q$
		to its target vertex (horizontal arrow) is blocked by vertices
		already holding their desired tokens (black circles).
		In any solution, $q$ takes a detour (curved arrow)
		or satisfied tokens at blocking vertices are moved.
		Either alternative results in additional swaps.}
	\label{fig:blocking}
\end{figure}
\begin{proof}
	Consider a solution $ \mathcal{S} = \tau_1, \tau_2, \ldots, \tau_N $
	and denote $ [N] \coloneqq \set{1, \ldots, N} $.
	We partition $ Q^* = Q^*_1\, \dot{\cup}\, Q^*_2 $
	according to the following scheme:
	for $ q \in Q^* $ let $ q \in Q^*_1 $
	if the path of $q$ in $ \mathcal{S} $
	contains all associated blocking vertices $ B_q $;
	otherwise, let $ q \in Q^*_2 $.
	
	For $ q \in Q_2^* $ and $ v, v' \in V $
	let $ \mathcal{P}_q(v, v') $ denote the set of $v$-$v'$-paths in $H$
	which do not contain all of the blocking vertices $ B_q $.
	For a path $ p \in \mathcal{P}_q(v, v') $,
	define its length as $ l(p) \coloneqq \card{p} - 1 $.
	Moreover, define
	\begin{align}
		d_{H, q}(v, v') \coloneqq \min \set{l(p) \mid \in \mathcal{P}_q(v, v')}
	\end{align}
	as the minimum edge-distance between $v$ and $v'$
	if only paths not containing all $q$-blocking vertices are allowed.
	Note, that $ \mathcal{P}_q(a(q), a'(q)) \neq \emptyset $:
	from the definition of $ Q^*_2 $,
	it follows that for all $ q \in Q^*_2 $ there is a path not containing~$ B_q $.
	Furthermore, by property (i) of $ B_q $ it holds 
	\begin{equation}\label{eq:detour}
		d_{H, q}(a(q), a'(q)) \geq d_H(a(q), a'(q)) + 2k_q,
	\end{equation}
	for all $ q \in Q^*_2 $,
	see Figure~\ref{fig:blocking} for an illustration.
	
	For $ i \in [N] $, consider the sum $ d_i $
	of all edge distances from tokens \textcolor{blue}{in $Q^*$} to their targets
	after swap $ \tau_i $ has been performed,
	accounting for the detours taken by tokens in $ Q^*_2 $,
	\begin{align*}
		d_i \coloneqq &\sum_{q \in Q^*_1} d_H(\tau_i\circ \tau_{i - 1} \circ \ldots \circ \tau_1 \circ a(q), a'(q)) +\\& \sum_{q \in Q^*_2} d_{H, q}(\tau_i\circ \tau_{i - 1} \circ \dots \circ \tau_1 \circ a(q), a'(q)).
	\end{align*}
	Accordingly, set
	\begin{align}
		d_0 \coloneqq \sum_{q \in Q_1^*} d_H(a(q), a'(q))
			+ \sum_{q \in Q^*_2} d_{H, q}(a(q), a'(q)).
	\end{align}
	Using~\eqref{eq:detour} yields
	\begin{equation}
		\label{eq:d0}
		d_0 \geq \sum_{q \in Q^*} d_H(a(q), a'(q)) + \sum_{q \in Q^*_2} 2k_q
			= \underbrace{\sum_{q \in Q} d_H(a(q), a'(q))}_{\eqqcolon D_0}
				+ \sum_{q \in Q^*_2} 2k_q,
	\end{equation}
	where we used that $ d_H(a(q), a'(q)) = 0 $
	for all $ q \in Q \setminus Q^* $.
	Furthermore, for $ i \in [N] $ let $ \Delta d_i $ denote the difference
	\begin{align}
		\Delta d_i \coloneqq d_i - d_{i - 1}.
	\end{align}
	Clearly, $ d_N = 0 $,
	since $ \tau_N \circ \ldots \circ \tau_2 \circ \tau_1 \circ a = a' $.
	Thus,
	\begin{equation}
		\label{eq:sumdelta}
		\sum_{i \in [N]} \Delta d_i = d_N - d_0 = -d_0.
	\end{equation}
	Furthermore, for any $ i \in [N] $ it holds $ -2 \leq \Delta d_i \leq 2 $,
	since a swap moves two tokens each for a distance of~1.
	
	Now, we turn to $ Q_1^* $.
	By definition of $ Q_1^* $,
	a vertex $ v \in \bigcup_{q \in Q^*_1} B_q $
	is necessarily involved in a swap.
	For $ v \in  \bigcup_{q \in Q^*_1} B_q $,
	denote by $ i_v \coloneqq \min \set{i \in[N] \mid v \in \tau_i} $
	the first swap that involves $i$.
	Let $ I \coloneqq \set{i_v \mid v \in \bigcup_{q \in Q^*_1} B_q} $.
	Partition $ I = I_0 \,\dot{\cup}\, I_2 $ in the following way.
	For $ i_v \in I $, if there is an $ v' \in \bigcup_{q \in Q^*_1} B_q $
	such that $v'\neq v$ but $i_v = i_{v'} $,
	we let $ i_v $ belong to the set~$ I_2 $.
	Otherwise, we let $ i_v $ belong to the set $ I_0 $.
	Then, for $ i \in I_0 $ we have $ \Delta d_i \geq 0 $,
	since at least one token leaves its target vertex by swap $ \tau_i $, namely the token $a^{-1}(v)$.
	Furthermore, for $ i \in I_2 $ we have $ \Delta d_i = 2 $,
	since both tokens leave their target vertices by swap $ \tau_i $.
	
	Let $ n_0 \coloneqq \card{I_0} $ and $ n_2 \coloneqq \card{I_2} $.
	Using property (ii) of $ B_q $ and the definition of $ I_0 $ and $ I_2 $,
	it holds
	\begin{align}
		\card{\bigcup_{q \in Q^*_1} B_q} = \sum_{q \in Q^*_1} k_q = 2n_2 + n_0.
	\end{align}
	Let $ \bar{I} \coloneqq [N] \setminus{I} = [N] \setminus(I_0 \cup I_2) $.
	Then we can use~\eqref{eq:sumdelta} to obtain
	\begin{align}
		-d_0 = \sum_{i \in [N]} \Delta d_i
			&=\underbrace{\sum_{i \in I_0} \Delta d_i}_{\geq 0}
				+ \underbrace{\sum_{i \in I_2} \Delta d_i}_{= 2n_2}
				+ \sum_{i \in \bar{I}} \Delta d_i
			\geq 2n_2 + \sum_{i \in \bar{I}} \Delta d_i.
	\end{align}
	Thus, with \eqref{eq:d0}, we have
	\begin{align}
		\sum_{i \in \bar{I}} \Delta d_i \leq - D_0 - \sum_{q \in Q^*_2} 2k_q - 2n_2.
	\end{align}
	Since $ \Delta d_i \geq -2 $, it follows
	\begin{align}
		\card{\bar{I}} \geq \ceil*{\frac{D_0 + \sum_{q\in Q^*_2}2 k_q + 2n_2}{2}}
			= \ceil*{\frac{D_0}{2}} + \sum_{q\in Q^*_2} k_q + n_2.
	\end{align}
	Altogether, this means
	\begin{align*}
		N &= \card{I_0} + \card{I_2} + \card{\bar{I}}
			\geq n_0 + n_2 + \ceil*{\frac{D_0}{2}}
				+ \sum_{q\in Q^*_2} k_q + n_2\\
			&= \ceil*{\frac{D_0}{2}} + \sum_{q \in Q^*_2} k_q
				+ \sum_{q \in Q^*_1} k_q
			= \ceil*{\frac{D_0}{2}} + \sum_{q\in Q^*} k_q \ .
	\end{align*}
\end{proof}
\textcolor{blue}{
The question remains, how the sets of blocking vertices $B_q$ can be computed in practise.
We employ a simple greedy procedure that iterates over all tokens and takes as many blocking vertices as possible in each iteration.}
As it turns out computationally,
this lower bound is particularly strong on graphs with low connectivity.
On the contrary, the second lower bound
employed by the branch-and-bound algorithm
is stronger on dense graphs.

\subsubsection{Complete Split Graph Lower Bound}

While the bound from Lemma~\ref{lem:distlb}
is based on relaxing the restriction,
that tokens need to be moved by swaps,
we now derive a lower bound
that is based on relaxing the set of allowed transpositions~$ \TT_H $.
This set is increased by adding edges missing in~$H$.
We use a result of Yasui et al.\ (\cite{Yasui2015SwappingLT}),
who give an efficient algorithm for solving token swapping problems
on \emph{complete split graphs}.
\begin{definition}[Complete split graph]
	A vertex subset of a graph is called \emph{independent}
	if no two vertices in the subset are connected by an edge.
	A vertex subset is called a \emph{clique}
	if any two vertices in the subset are connected by an edge.
	A graph is called a \emph{complete split graph}
	if it can be partitioned into an independent set and a clique
	such that every vertex in the independent set
	is connected to all vertices in the clique. 
\end{definition}
\begin{lemma}[Complete split graph lower bound \cite{Yasui2015SwappingLT}]
	\label{lem:competesplit}
	Let $ T = (H, Q, a, a') $ be a token swapping instance.
	Furthermore, let $ V' \subset V $ be an independent set in $H$.
	Let $ n = \card{V} $ be the number of vertices,
	let~$r$ be the number of cycles in the permutation $ a \circ (a')^{-1} $
	including fix points, \ie cycles of length 1.
	Let $q$ be the number of cycles in the permutation $ a \circ (a')^{-1} $
	excluding fix points that only involve vertices in $ V' $.
	Then it holds
	\begin{align}
		\mathrm{OPT}(T) \geq n - r + 2q.
	\end{align}
\end{lemma}
\begin{proof}
	Consider the graph $ H' = (V, E') $
	which is derived from $H$ and $ V' $ in the following way:
	connect every vertex in $ V' $ with all vertices in $ V \setminus V' $.
	Additionally, add all missing edges in $ V \setminus V' $.
	Then $ H' $ is a complete split graph.
	Furthermore, $ E \subseteq E' $ holds,
	and thus also $ \TT_H \subseteq \TT_H' $.
	Consider a solution $ \mS $ of size $N$ to $T$.
	This is also a solution
	to the token swapping instance $ T' \coloneqq (H', Q, a \circ (a')^{-1}, \mathrm{id}) $.
	Thus, we have
	\begin{align}
		\card{\mS} \geq \mathrm{OPT}(T').
	\end{align}
	According to \cite{Yasui2015SwappingLT},
	the optimal value for $T'$ is given by
	\begin{align}
		\mathrm{OPT}(T') = n - r + 2q.
	\end{align}
	Thus,
	\begin{align}
		\card{\mS} \geq n - r + 2q.
	\end{align}	
\end{proof}
We remark that Lemma~\ref{lem:competesplit}
is a generalization of the well-known fact that token swapping
on complete graphs has optimal value $ n - r $ (\cf \eg \cite{Kim2016}).
\textcolor{blue}{
When applying Lemma~\ref{lem:competesplit} in practise, an independent set needs to be constructed first.
In our implementation, we employ a simple greedy procedure that iterates over the vertices increasing in degree, adding it to the independent set if possible.
}

As already mentioned, both lower bounds complement each other:
while the distance bound is strong on sparse graphs,
the complete-split-graph bound is strong on dense graphs.
The following lemma allows to increase the lower bounds
from Lemmas~\ref{lem:distlb} and~\ref{lem:competesplit}
by one in some cases.
\begin{lemma}[Parity of solutions]
	\label{lem:parity}
	Let $ T = (H, Q, a, a') $ be a token swapping instance.
	Let $ \mathrm{sgn}(\pi) $ denote the sign of a permutation~$ \pi $.
	For any solution $ \mS = \tau_1, \tau_2, \ldots, \tau_N $ to $T$,
	it holds
	\begin{align}
		N \equiv \left[1 - \mathrm{sgn}(a' \circ a^{-1})\right] / 2 \mod 2
	\end{align}
\end{lemma}
\begin{proof}
	For any solution $ \mS $, it holds
	\begin{align}
		a' = \tau_N \circ \ldots \circ \tau_2 \circ \tau_1 \circ a
	\end{align}
	or, equivalently
	\begin{align}
		a' \circ a^{-1} = \tau_N \circ \ldots \circ \tau_2 \circ \tau_1.
	\end{align}
	Using the fact that $ \mathrm{sgn}(\sigma \circ \pi)
		= \mathrm{sgn}(\sigma) \cdot \mathrm{sgn}(\pi) $
	for any two permutations $ \sigma, \pi $, we have
	\begin{align*}
		\mathrm{sgn}(a' \circ a^{-1})
			= \mathrm{sgn}(\tau_N) \cdots \mathrm{sgn}(\tau_2)
				\cdot \mathrm{sgn}(\tau_1)
			= (-1)^N.
	\end{align*}
	Thus, $N$ is either even or odd, independent of $S$:
	\begin{align}
		N \equiv \left[1 - \mathrm{sgn}(a' \circ a^{-1})\right] / 2 \mod 2.
	\end{align}
\end{proof}
With these theoretical properties of our solution method
to the qubit routing by swap insertion problem established,
we now present various experimental results.

	\section{Experimental Results}
\label{sec:exp}
In the following, we benchmark the proposed routing algorithm.
We perform experiments covering four aspects:
\begin{enumerate}
	\item The computation time and solution quality
		when solving the TAP model with strengthening inequalities.
	\item The quality and speed of the token swapping approximation algorithm.
	\item The quality of the routing solution produced by the combined algorithm
		compared to other routing methods.
	\item The quality of solutions produced by quantum algorithms compiled with our routing method,
		executed on actual quantum hardware.
\end{enumerate}
The datasets generated and analyzed during the current study are
available from the corresponding author on reasonable request.

\subsection{Benchmarking the Strengthened TAP Model}\label{sec:benchsgi}
In the first step of the proposed decomposition approach, a TAP needs to be solved.
Here, we give numerical results
showing how the introduction of subgraph isomorphism cutting planes
(Section~\ref{sec:subgi}) can
accelerate the solution of the TAP binary linear program.
The \textcolor{blue}{linearized} flow model~\eqref{eq:edgemodel} is implemented and solved
via the Python API of \emph{Gurobi} (\cite{gurobi}).
Before parsing the model to the solver,
we generate a set of valid inequalities introduced in Section~\ref{sec:subgi}.
These are added to a cut pool.
The solver can add cuts from this pool to the model
in order to cut off relaxed solutions.

The set of valid inequalities is generated as follows:
first, the graphs $ H^d $ for $ 0 \leq d \leq \max_{i, j \in H} d_H(i, j) $
are created and transformed into their respective line graphs.
To generate a subset of gates for which a valid cut may be derived,
we first choose time steps~$ t_0 $ and~$ t_1 $ such that $ 1 \leq t_0 < t_1 \leq L $.
Then, all gates in gate groups between~$ t_0 $ and~$ t_1 $ are considered.
The corresponding connectivity graph is generated
and transformed into its line graph.
To perform the node-induced-subgraph isomorphism check,
we employ an implementation of the VF2 algorithm
(see \cite{cordella2005isom, cordella01isom})
from the NetworkX python package (\cite{hagbger2008networkx}).
This procedure is repeated by iterating over time steps $ t_0 $, $ t_1 $
until all possible \textcolor{blue}{time step pairs} have been checked.

In the first benchmark, we measure the computation time
required for solving TAP instances to optimality.
We compare the total time taken with strengthening inequalities
to the time taken for the bare model~\eqref{eq:edgemodel}.
We measure the ratio $ r_t = t_{\mathrm{SGI}} / t_{\mathrm{bare}} $,
where $ t_{\mathrm{SGI}} $ and $ t_{\mathrm{bare}} $
are the runtimes of the strengthened model and the bare model, respectively.
\textcolor{blue}{The time consumed by cutting plane generation
lies between $0.1$~s and 20~s and}
is included in the total runtime $ t_{\mathrm{SGI}} $.

\textcolor{blue}{In all instances $ (H, Q, \Gamma) $,
the hardware graph $H$ has 8 vertices
and $ Q \coloneqq \set{q_1 \ldots, q_8} $.
Tests are performed on four hardware graphs with different sparsity,
namely the line, ring, Y and ladder graph, illustrated in Fig~\ref{fig:Hs}.}
Each of the gate groups in $ \Gamma = G^1, G^2, \ldots, G^L $
consists of $\floor*{L/2}$ randomly chosen pairs of qubits among the set $\set{q_1,\cdots,q_{L}}$.
This results in instances with as many gates per time step as possible,
making the instances computationally hard.
The depth~$L$ is varied from 4 to 8,
where 10 instances are generated per value of $L$.
The instance sizes resemble the capabilities
of currently available quantum hardware.
All instances are solved to optimality.
\textcolor{blue}{The absolute solutions times lie between 0.6~s and 1000~s.
Results are given in Table~\ref{tab:allocrt}.
Shorter average runtimes are achieved for
$L=5$ on ring, Y, and ladder.
A maximum saving of to 72\% is observed for the $L=4$ line.
Furthermore, we observe rather high differences between minimum and maximum runtime ratio,
indicating a strong instance dependence.
Summarizing, the separation of valid subgraph isomorphism inequalities
can improve the TAP-model runtime.
Thus, also the algorithm for the full routing problem
benefits from generating valid inequalities for the TAP subproblem.
However, the strong instance dependence
suggests that more research is necessary in order to fully exploit the potential
of subgraph isomorphism inequalities.}

For the same benchmark set,
we measure the tightness of the lower bound
on the number of required swaps in an optimal routing solution,
given by the optimal TAP solution.
To this end, the objective of the full routing model (BIP)
from \cite{nannicini2021} is adjusted to minimize swap count instead of depth.
Solving this model yields the optimal swap number.
We measure the relative bound \textcolor{blue}{$ r_b = \frac{\mathrm{OPT}_{\mathrm{TAP}}}{\mathrm{OPT}_{\mathrm{BIP}}} $},
with $ \mathrm{OPT}_{\mathrm{BIP}} $ and $ \mathrm{OPT}_ {\mathrm{TAP}} $
denoting the optimal BIP solution and TAP solution, respectively.
These results are shown in Table~\ref{tab:allocb}.
We observe tight bounds with a relative difference of at most 23\% ($L=8$ line),
together with small deviations from the average value.
This is a key feature of our decomposition approach:
optimal TAP solution values typically give a good estimate
for the number of required swaps.
\begin{figure}[]
	\centering
	\subfloat[Line]{
		\includegraphics[height=2cm]{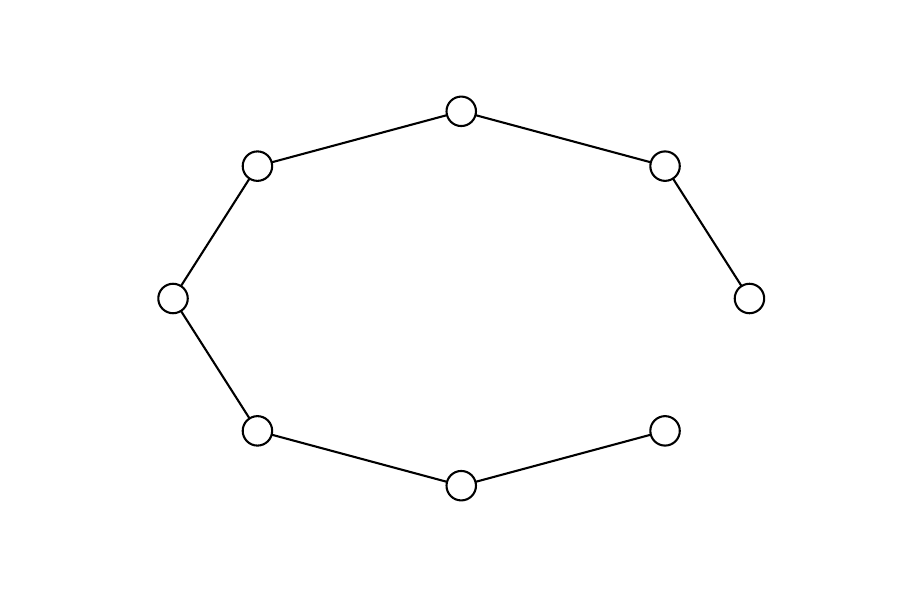}
		\label{fig:Hline}
	}
	\subfloat[Ring]{
		\includegraphics[height=2cm]{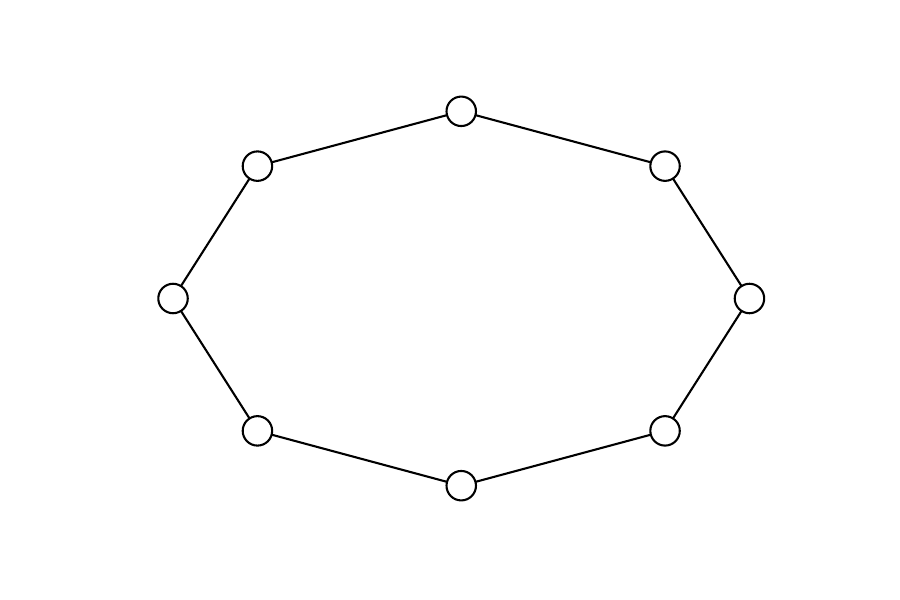}
		\label{fig::Hline}
	}
	\subfloat[Y]{
	\includegraphics[height=2cm]{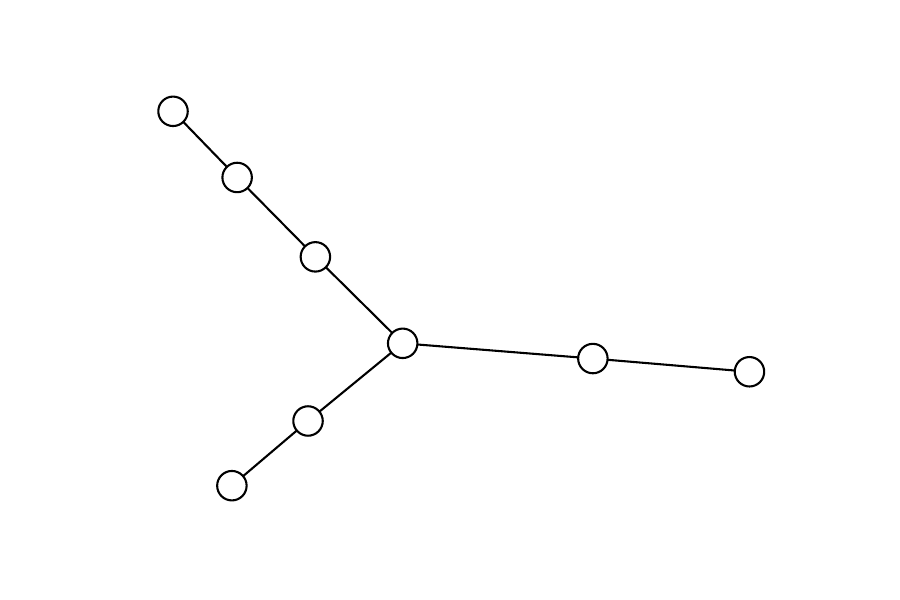}
	\label{fig::Hy}
	}
	\subfloat[Ladder]{
	\includegraphics[height=2cm]{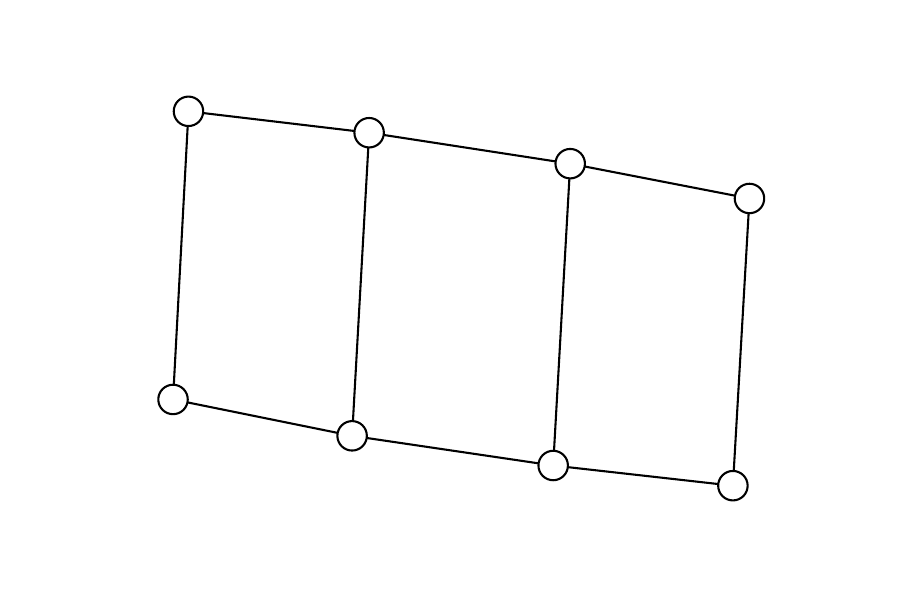}
	\label{fig::Hladder}
	}
	\caption{Hardware graphs used in the benchmarks of Sec.~\ref{sec:benchsgi} and Sec.~\ref{sec:bench}.}
	\label{fig:Hs}
\end{figure}
\begin{table}
	\makebox[\textwidth][c]{
	\subfloat[Rutnime comparison when solving to optimality.]
	{
		\begin{tabular}{cccccccccccccccc}
			\toprule
			&\multicolumn{15}{c}{$L$} \\
			\cmidrule{2-16}
			& \multicolumn{3}{c}{4} & \multicolumn{3}{c}{5} & \multicolumn{3}{c}{6} & \multicolumn{3}{c}{7} & \multicolumn{3}{c}{8}\\
			\cmidrule(lr){2-4}\cmidrule(lr){5-7}\cmidrule(lr){8-10}\cmidrule(lr){11-13}\cmidrule(lr){14-16}
			$H$ & min & max & avg & min & max & avg & min & max & avg & min & max & avg & min & max & avg  \\
			\midrule
			Line & 0.28 & 3.9 & 1.31 & 0.8 & 2.56 & 1.28 & 1.27 & 7.76 & 2.96 & 0.31 & 4.19 & 2.49 & 0.47 & 2.35 & 1.35 \\
			\midrule
			Ring & 0.5 & 0.98 & 1.02 & 0.64 & 1.9 & 0.8 & 0.76 & 3.72 & 1.56 & 0.74 & 2.78 & 1.34 & 0.48 & 1.47 & 1.01 \\
			\midrule
			Y & 0.48 & 1.58 & 1.16 & 0.49 & 2.76 & 0.87 & 0.75 & 2.54 & 1.49 & 0.77 & 3.39 & 1.89 & 0.56 & 2.87 & 1.74 \\
			\midrule
			Ladder & 0.36 & 1.16 & 1.2 & 0.48 & 2.47 & 0.9 & 0.78 & 9.72 & 2.62 & 0.63 & 4.58 & 2.4 & 0.93 & 7.26 & 2.53 \\
			\bottomrule
		\end{tabular}
		\label{tab:allocrt}
	}}\\
	\makebox[\textwidth][c]{
	\subfloat[Bound quality when solving to optimality.]
	{
		\begin{tabular}{cccccccccccccccc}
			\toprule
			&\multicolumn{15}{c}{$L$} \\
			\cmidrule{2-16}
			& \multicolumn{3}{c}{4} & \multicolumn{3}{c}{5} & \multicolumn{3}{c}{6} & \multicolumn{3}{c}{7} & \multicolumn{3}{c}{8}\\
			\cmidrule(lr){2-4}\cmidrule(lr){5-7}\cmidrule(lr){8-10}\cmidrule(lr){11-13}\cmidrule(lr){14-16}
			$H$ & min & max & avg & min & max & avg & min & max & avg & min & max & avg & min & max & avg  \\
			\midrule
			Line & 1.0 & 1.0 & 1.0 & 0.8 & 1.0 & 0.98 & 0.9 & 1.0 & 0.98 & 0.85 & 1.0 & 0.95 & 0.77 & 1.0 & 0.86 \\
			\midrule
			Ring & 1.0 & 1.0 & 1.0 & 0.8 & 1.0 & 0.98 & 1.0 & 1.0 & 1.0 & 1.0 & 1.0 & 1.0 & 0.92 & 1.0 & 0.97 \\
			\midrule
			Y & 1.0 & 1.0 & 1.0 & 0.8 & 1.0 & 0.94 & 0.89 & 1.0 & 0.96 & 0.85 & 1.0 & 0.93 & 0.81 & 1.0 & 0.89 \\
			\midrule
			Ladder & 1.0 & 1.0 & 1.0 & 1.0 & 1.0 & 1.0 & 1.0 & 1.0 & 1.0 & 0.8 & 1.0 & 0.94 & 0.83 & 1.0 & 0.97 \\
			\bottomrule
		\end{tabular}
		\label{tab:allocb}
	}}\\
	\caption{\textcolor{blue}{TAP-model benchmark results.
		The different hardware graphs, denoted $H$, are depicted in Fig.~\ref{fig:Hs}.
		The qubit set is $ Q \coloneqq \set{q_1 \ldots, q_8} $.
		Each instance contains $L$ many gate groups $ G^1, G^2, \ldots, G^L $,
		where each gate group consists of $\floor*{L/2}$ randomly chosen pairs of qubits among the set $\set{q_1,\cdots,q_{L}}$.
		10 instances are generated per depth $L$ and graph $H$.
		Each instance is solved
		with and without additional strengthening inequalities.}
		We compare the ratio $ r_t $ of runtimes in~(a).
		Additionally, in~(b)
		we give the relative bound quality $ r_b $ for the optimal TAP value
		as a lower bound on the number of required swaps in an optimal routing solution.}
	\label{tab:allocbm}
\end{table}

\subsection{Benchmarking the Token Swapping Approximation}
The second step in the decomposition approach
consists of several consecutive token swapping instances.
Here, we evaluate the performance
of the employed token swapping approximation.
All algorithms described in this section
have been implemented in Python from scratch.
First, we compare the original 4-approximation algorithm
given in \cite{miltzow2016} to our modified version.
For this purpose, we generate random token swapping instances,
\ie initial and final allocations\footnote{Due to symmetry, it is sufficient to only vary initial allocations while fixing the final allocation.}
on three types of planar graphs with different {sparsity}.
As graph types we choose ring graphs, ladders and square meshes with node numbers of 16, 36 and 64
for each type.
As in the previous section, this choice of graph types and sizes
aims at resembling current realizations
of quantum computer hardware connectivity graphs.
We generate 100 instances per graph.
In Table~\ref{tab:tsapp}, we give the average ratio $ r_\# $ of swap numbers
between the modified version and the original version
as well as the average ratio of depths $ r_d $ .
\begin{table}
	\centering
	\begin{tabular}{ccccccc}
		\toprule
		&\multicolumn{2}{c}{Ring} & \multicolumn{2}{c}{Ladder}& \multicolumn{2}{c}{Mesh} \\
		\cmidrule(r){2-3}\cmidrule(rl){4-5}\cmidrule(l){6-7}
		$N$ & $r_\#$ & $r_d$ & $r_\#$ & $r_d$ & $r_\#$ & $r_d$ \\
		\midrule
		16 & 0.96 & 0.91 & 0.88 & 0.82 & 0.81& 0.78\\
		\midrule
		36 & 1.0 & 1.0& 0.91 & 0.83 &  0.83 & 0.78 \\
		\midrule
		64 &  1.0 & 0.99 & 0.93 & 0.79 & 0.84 & 0.81\\
		\bottomrule
	\end{tabular}
	\caption{Comparison of original
		and modified token swapping approximation algorithms.
		We give the average ratio of inserted swaps $ r_\# = \#_{\text{mod}} / \#_{\text{orig}} $
		and resulting depths $ r_d = d_{\text{mod}} / d_{\text{orig}} $.
		Comparison are made on rings, ladders and meshes
		with increasing vertex number~$N$.}
	\label{tab:tsapp}
\end{table}

Firstly, note that these ratios are less than or equal to $1$
across all instances.
This means that our modification always performs as least as good
as the original version, or better (if the factor is smaller than $1$).
Furthermore, both swap and depth ratio
are roughly constant across graph sizes for a given graph type.
However, the advantage of the modified version increases
with growing connectivity of the graph type.
An intuitive explanation of this behaviour
is the increase of swapping possibilities with growing edge number.
This allows the modified version to employ its potential
of deciding for \qm{good} swaps instead of choosing them randomly.

Secondly, we compare the modified 4-approximation algorithm
to the exact branch-and-bound algorithm
in terms of runtime and solution quality.
We measure the approximation ratio $ r_\# = \#_{\text{approx}} / \#_{\text{opt}} $
as well as the runtime ratio $ r_t = t_{\text{approx}} / t_{\text{bnb}} $
on various instances.
As graph types, we choose lines, rings and ladders.
Again, this choice tries to resemble real hardware graphs
while we also want to make use of the fact
that token swapping can be solved efficiently on lines and rings\footnote{In fact, token swapping can be solved efficiently on lines, rings, cliques, stars and brooms, see \cite{Kim2016}.}.
This allows use to measure the approximation quality
for instances which are intractable for the exact algorithm.
Experimental results are summarized in Table~\ref{tab:tsexact}.

First, we observe that for lines
the approximation ratio is unity across all instances.
As it might be conjectured from these numerical results,
one can also convince oneself theoretically
that the approximation algorithm always finds an optimal solution on lines.
 
Furthermore, Table~\ref{tab:tsexact} shows
that for ring and ladder graphs
the average approximation ratio is less than or equal to 1.25,
which is well below the theoretical value of 4.
This indicates that the theoretical approximation factor of 4
underestimates the approximation quality observed in practice.
As already mentioned, the proof in \cite{miltzow2016} delivers an approximation ratio
of strictly less than four.

Comparing runtimes, we note a drastic improvement
with growing graph size on all graph types.
For the largest instances
the exact algorithm was able to solve in reasonable time,
an average reduction of 96.5\% is achieved.
Moreover, all absolute runtimes for the approximation algorithm
are well below 1~s. Altogether, the approximation algorithm typically delivers
close-to-optimal solutions
with a significant time saving on all instances considered.

\begin{table}
	\centering
	\begin{tabular}{ r r r r r r r r r r }
		\toprule
		&\multicolumn{3}{c}{Line} & \multicolumn{3}{c}{Ring}& \multicolumn{3}{c}{Ladder} \\
		\cmidrule(r){2-4}\cmidrule(rl){5-7}\cmidrule(l){8-10}
		$N$ & $r_\#$ & $r_t$ & $\bar{t}_{\text{appr}}$&$r_\#$& $r_t$ &$\bar{t}_{\text{appr}}$& $r_\#$ & $r_t$&$\bar{t}_{\text{appr}}$ \\
		\midrule
		5 & 1.0& 0.2 & 0.0& 1.25 & 0.3& 0.0 & 1.04 & 0.6& 0.0 \\
		\midrule
		10 & 1.0 & 0.05 & 0.0 & 1.19 & 0.2 & 0.0& 1.19 & 0.2& 0.0\\
		\midrule
		12 & 1.0 & - & 0.0 & 1.23 & 0.08& 0.0 & 1.13 & 0.035 & 0.0\\
		\midrule
		50 &  1.0 & - & 0.8 & 1.17 & - & 0.4 &  & & 0.4\\
		\bottomrule
	\end{tabular}
	\caption{Comparison of token swapping approximation and exact algorithm.
		We give the average approximation ratio $ r_\# =
			\#_{\text{appr}} / \#_{\text{opt}} $
		and runtime ratio $ r_t = t_{\text{appr}} / t_{\text{bnb}} $
		as well as the average absolute runtime
		of the approximation algorithm $ \bar{t}_{\text{appr}} $ in seconds.
		Instance graphs are lines, rings and ladders
		of increasing vertex number~$N$.
		Token swapping can be solved efficiently on lines and rings,
		\cf \cite{Kim2016}.
		This allows us to compare approximation ratios
		even when the runtime is intractable
		for the exact algorithm (indicated by \qm{-}).
		For ladders, no efficient exact algorithm exists.
		Therefore, no results are presented for large ladder graphs (empty cells).
		Averages are taken over 100 random instances.
		Exceptions due to excessive runtime are:
		54 instances for $ N = 10 $ lines, 22 for $ N = 12 $ rings
		and 18 for $ N = 12 $ ladders.}
	\label{tab:tsexact}
\end{table}

\subsection{Benchmarking the Routing Algorithm}\label{sec:bench}

In the following, the proposed routing method (abbreviated TAP+TS)
is compared to various routing algorithms
integrated in different open source quantum compilers:
The binary program (BIP) by Nannicini et al. (\cite{nannicini2021}),
solved via \emph{Gurobi},
Li et al.'s Sabre heuristic (\cite{Li2019}),
Qiskit's default compiler Stochastic swap (\cite{Qiskit}),
Cambridge Quantum Computing's Tket (\cite{cowtan2019,Sivarajah2020}),
Google's Cirq \cite{cirq_developers_2021_5182845},
\textcolor{blue}{
as well as the Bounded mapping tree (BMT) by Siraichi et.~al (\cite{Siraichi2019})
which is based on a decomposition approach similar to ours. 
BMT employs dynamic programming to solve a problem closely related to the TAP.}

We briefly summarize the Python implementation of our routing algorithm:
first, the two-qubit gates of the input circuit are grouped in layers containing as much gates as possible.
These gate groups together with the hardware graph form a TAP instance.
Next, the network flow model of the TAP instance is solved by \emph{Gurobi}.
The resulting TAP solution defines a set of token swapping problems,
which are then solved by the approximation algorithm
described in Section~\ref{sec:ts}.
Finally, the routed quantum circuit is constructed
by changing the target qubits in the gates of the original circuit
according to the allocation sequence
and by adding the computed swaps between subsequent allocations.

\textcolor{blue}{As performance metrics, we measure the resulting gate count and circuit depth.}
Benchmarks are performed on two different classes of circuits.
The first class is taken from a publicly available benchmark library
while the second is generated according to a protocol
used for benchmarking quantum hardware. 

\paragraph{QUEKO circuits.}

First, we use the publicly available benchmarks circuits
generated by Tan and Cong in \cite{Tan21}.
They propose a method for generating random circuits
on a given hardware graph with given depth
as well as given single and two-qubit gate numbers.
Their construction delivers instances of the routing by swap insertion problem
with optimal value $0$,
\ie there exists an allocation of logical to physical qubits
such that all gates can be performed without inserting swaps.

We use the \qm{benchmark for near-term feasibility} (BNTF) circuit set
from \cite{Tan21}.
It is subdivided into circuits on a 16-qubit hardware graph (Rigetti's Apsen-4, Fig.~\ref{fig:qeko16h})
with few gates, and circuits on a 54 qubit graph (Google’s Sycamore, Fig.~\ref{fig:qeko54h})
with many gates.
For each hardware graph, 90 circuits are available --
10 for each depth in $ \set{5, 10, \ldots, 45} $.
On each instance the time limit for BIP is the set to 3600~s.
The results are visualized in Figure~\ref{fig:resrouting}.
By construction, our algorithm will find a zero-swap solution
if the TAP is solved to optimality.
This can be observed in Figures~\ref{fig:qeko16g} and~\ref{fig:qeko54g}.
For this benchmark set, we achieve improvements in resulting depth
of up to a factor of~2 on small circuits (Figure~\ref{fig:qeko16d})
and up to factor of~5 on larger circuits (Figure~\ref{fig:qeko54d}),
compared to the best heuristic.
This improvement comes at the cost of increased runtime.
Indeed, the runtimes of all other heuristics applied in this benchmark
are in the order of a few seconds (not shown in Figure~\ref{fig:resrouting}),
while the proposed algorithm takes up to 12~s on average for small circuits (Figure~\ref{fig:qeko16t})
and up to 5~min on average for larger circuits (Figure~\ref{fig:qeko54t}).
However, we note a drastic reduction of runtime compared to the exact approach BIP (Figure~\ref{fig:qeko16t})
which is unable to find optimal solutions on larger instances in less than one hour.
\begin{figure}
	\centering
	\subfloat[]{\includegraphics[width=0.3\linewidth]{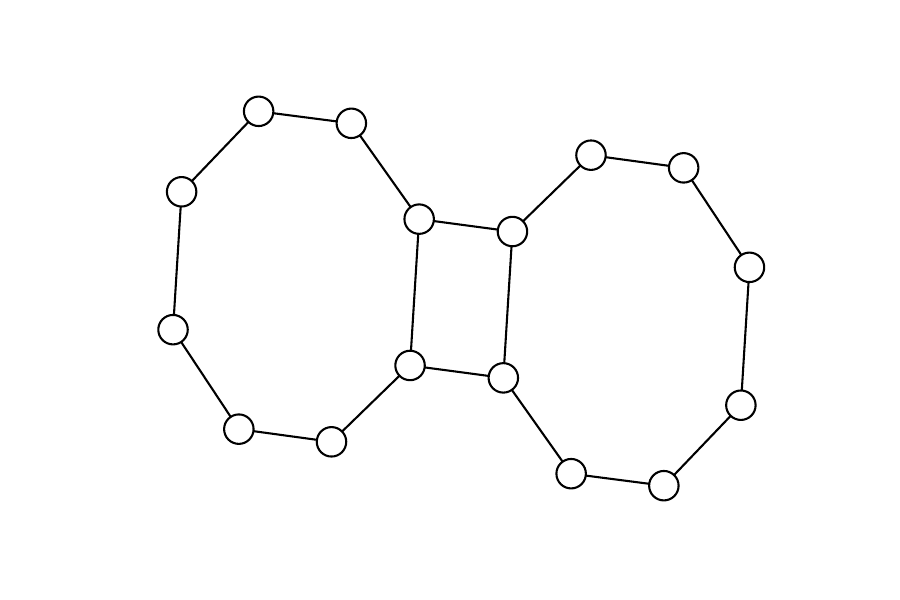}\label{fig:qeko16h}}
	\subfloat[]{\includegraphics[width=0.3\linewidth]{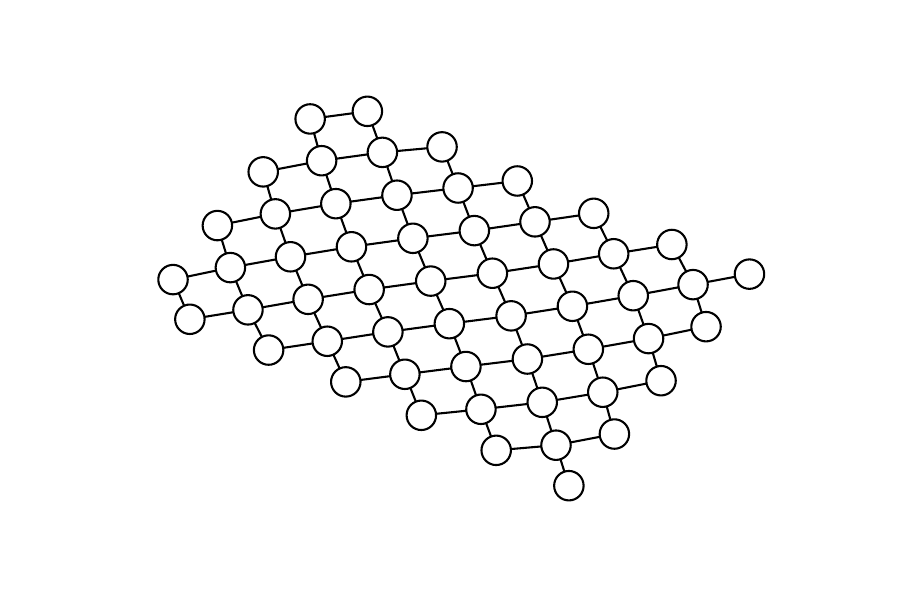}\label{fig:qeko54h}}
	\subfloat[]{\includegraphics[width=0.3\linewidth]{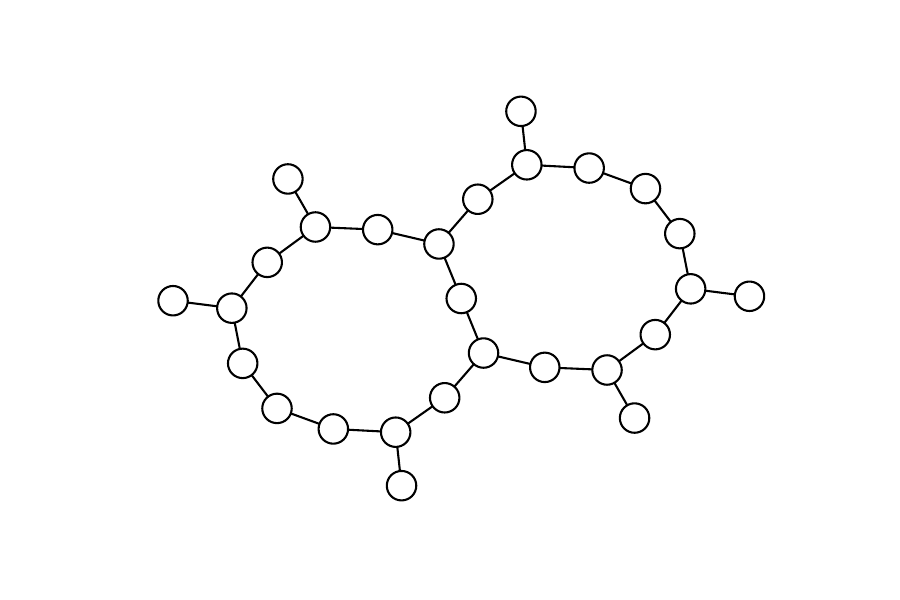}\label{fig:ehnh}}
	\caption{\textcolor{blue}{Hardware graphs of \emph{Aspen-4} (a), \emph{Sycamore} (b) and \emph{ibmq\_ehningen} (c).}
	}
	\label{fig:quekoh}
\end{figure}
\begin{figure}
	\centering
	\subfloat[Gatecount comparison on Aspen-4.]{\includegraphics[width=0.49\linewidth]{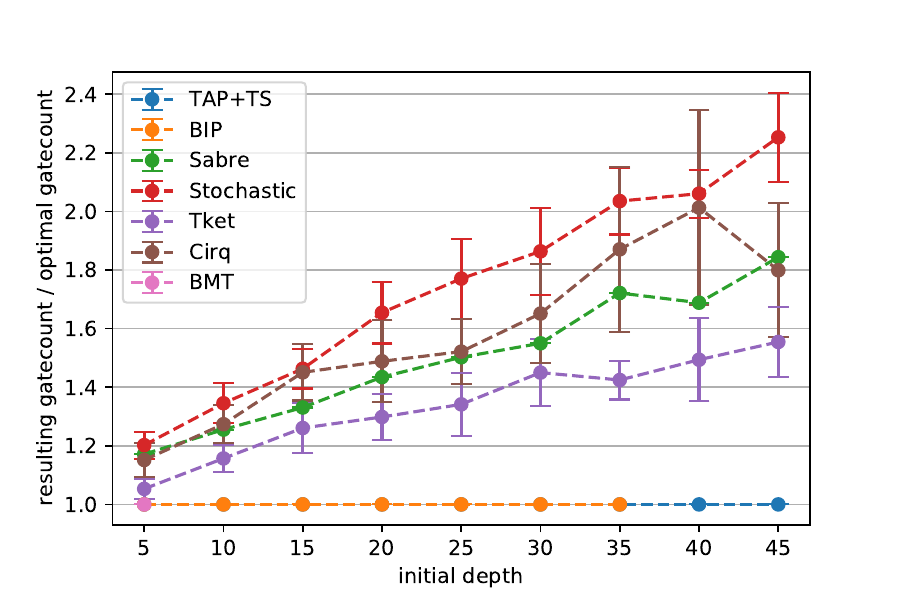}\label{fig:qeko16g}}
	\subfloat[Gatecount comparison on Sycamore.]{\includegraphics[width=0.49\linewidth]{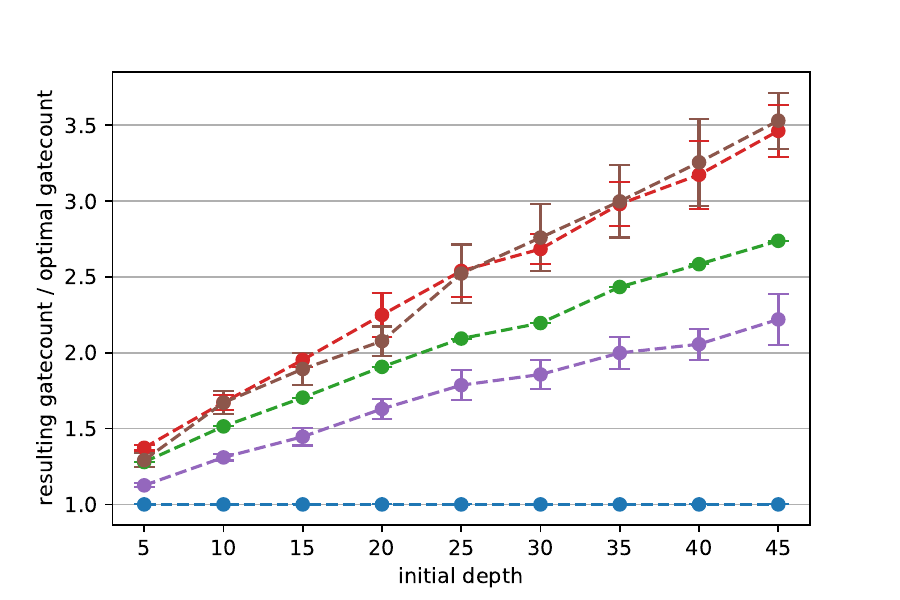}\label{fig:qeko54g}}\\
	\subfloat[Depth comparison on Aspen-4.]{\includegraphics[width=0.49\linewidth]{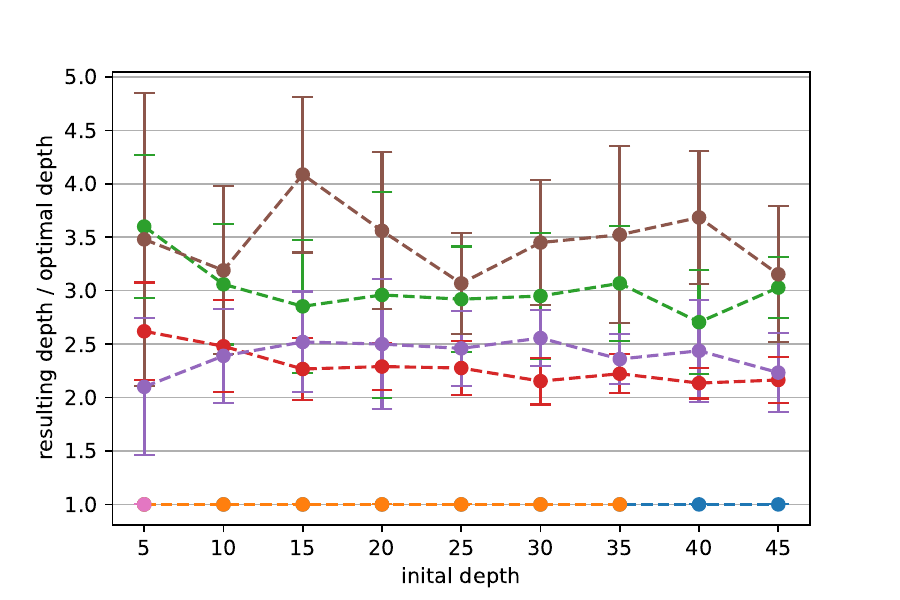}\label{fig:qeko16d}}
	\subfloat[Depth comparison on Sycamore.]{\includegraphics[width=0.49\linewidth]{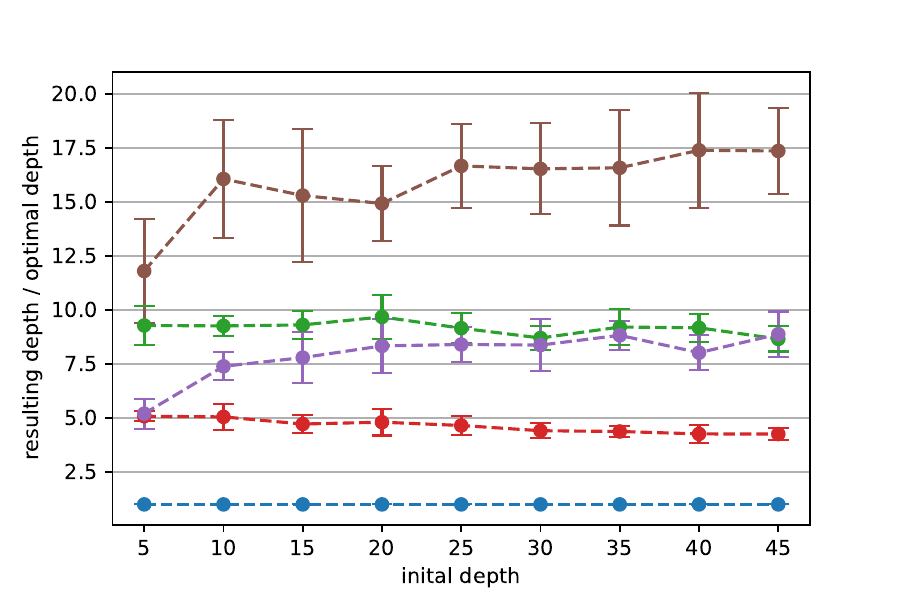}\label{fig:qeko54d}}\\
	\subfloat[Runtime comparison on Aspen-4.]{\includegraphics[width=0.49\linewidth]{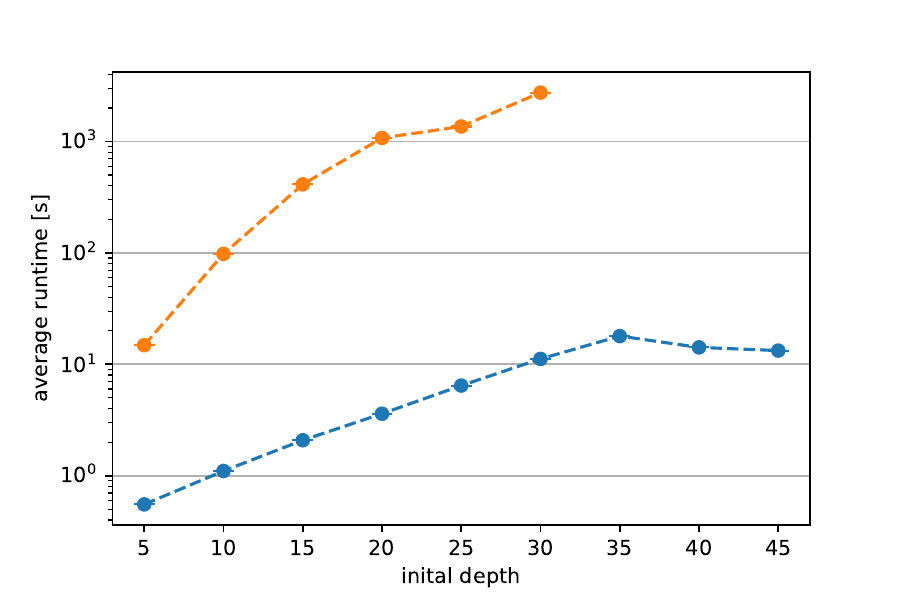}\label{fig:qeko16t}}
	\subfloat[Runtime comparison on Sycamore.]{\includegraphics[width=0.49\linewidth]{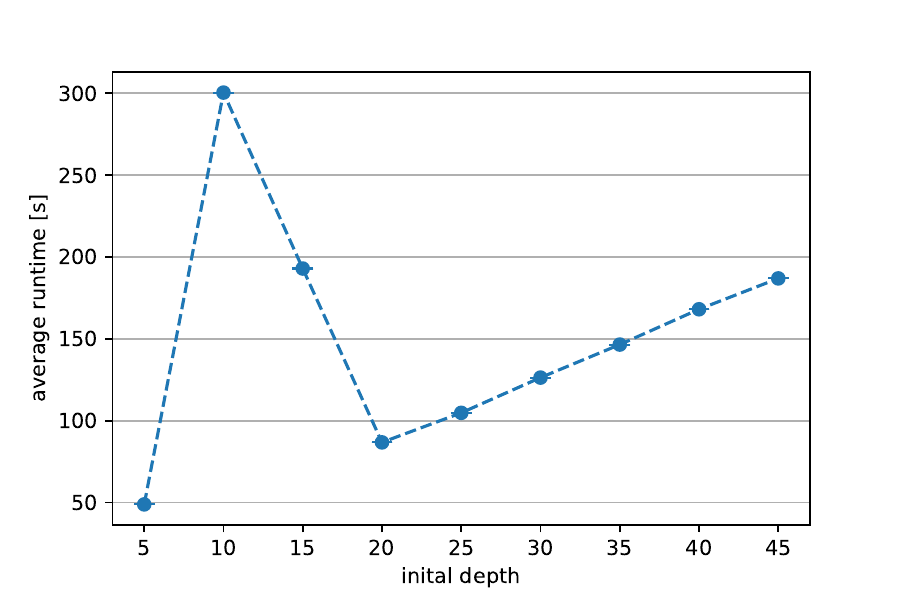}\label{fig:qeko54t}}\\
	\caption{\textcolor{blue}{Comparison of different routing algorithms on the QUEKO benchmark set.
		Data points are averaged over 10~instances.
		Error bars indicate the empirical standard deviation.
		The legend in (a) holds for all plots.
		Instances are the BNTF circuits from \cite{Tan21}
		on Rigetti's \emph{Aspen-4} 16-qubit hardware graph (left plots)
		as well as Google’s \emph{Sycamore} 54-qubit hardware graph (right plots).
		For every BNTF instance, a zero-swap solution exists.
		Missing data points for BIP and BMT indicate that no solution was returned after 3600~s runtime.
		(a), (b): Gatecount relative to optimum.
		(c), (d): Depth relative to optimum.
		(e), (f): Average runtime for BIP and BMT.
		Other heuristics are not shown since all runtimes are in the order of 1~s.}
		}
	\label{fig:resrouting}
\end{figure}

\paragraph{QV circuits.}

For a second benchmark,
we consider quantum volume (QV) circuits.
According to \cite{Moll2018}, a QV circuit of depth~$L$ and width~$m$
is a circuit on $m$~qubits with $L$~layers
where each layer consists of $ \floor{m / 2} $ random two-qubit gates.
Thus, the resulting circuit is maximally dense in terms of gates per layer.
QV circuits of equal depth and width are used to benchmark quantum computers
(see \cite{Cross2019}).
We consider QV circuits of equal depth and width.
\textcolor{blue}{
Circuits are constructed from the gate groups of the TAP instances generated in Sec~\ref{sec:benchsgi}
resulting in QV circuits with depths between $L=4$ and $L=8$.
As in Sec.~\ref{sec:benchsgi}, tests are performed on a 8-vertex-line, ring, Y and ladder graph.
Thus, the TAP instances associated with the routing instances generated here are exactly those of Sec.~\ref{sec:benchsgi}.}

\textcolor{blue}{
Optimal-depth solutions cannot be calculated in reasonable time for all instances.
However, BIP would find optimal solutions in terms of depth when solved to optimality.
Thus, we measure resulting depth and gate count relative to the BIP solution with a runtime limit of 7200~s.
}
Our results are visualized in Figures~\ref{fig:resrouting_qv} and \ref{fig:resrouting_qv2}.
In Figures~\ref{fig:qvline_d}, \ref{fig:qvring_d}, \ref{fig:qvy_d} and \ref{fig:qvl_d} we observe that the proposed algorithm
finds solutions which are close-to-optimal in terms of depth,
\textcolor{blue}{assuming BIP returned an optimal-depth solution after a runtime of 7200~s.
The average ratio of TAP+TS solution to BIP solution is at most 1.25\%.}
Furthermore, it performs best among all considered routing heuristics,
\textcolor{blue}{in particular better than the conceptually similar method BMT.}
This also holds true when comparing gate numbers,
as can be seen in Figures~\ref{fig:qvline_g}, \ref{fig:qvring_g}, \ref{fig:qvy_g} and \ref{fig:qvl_g}.
\textcolor{blue}{Interestingly, we observe that the proposed method delivers smaller gate numbers than BIP
for some instances on all topologies,
as indicated by values less than~1 in Figures~\ref{fig:qvline_g}, \ref{fig:qvring_g}, \ref{fig:qvy_g} and \ref{fig:qvl_g}.}
Again, the improvement comes at the cost of increased runtime.
While the other heuristics require several seconds
(not shown in Figs.~\ref{fig:resrouting_qv},\ref{fig:resrouting_qv2}),
\textcolor{blue}{the proposed method takes up to 900~s on average (Figure~\ref{fig:qvring_t}).
Still, the runtime is significantly shorter than for the exact method BIP 
which runs into the time limit of 7200~s on large instances (Figures~\ref{fig:qvring_t},~\ref{fig:qvl_t}).}

We discuss three reasons why the improved routing solutions
can be worth the additional time investment.
First, access to currently available quantum hardware
is often organized by a waiting queue,
such that the time for compilation is not a critical factor.
Second, different quantum circuits with equal connectivity graphs
only need to be routed once.
A popular example are parameterized quantum circuits
which are often considered as promising candidates
for first quantum applications.
The connectivity graph of such a circuit is not parameter dependent.
Third, current quantum hardware benefits strongly
from good routing solutions.
This is demonstrated in the next section.

\begin{figure}
	\centering
	\subfloat[Gatecount comparison on line.]{\includegraphics[width=0.49\linewidth]{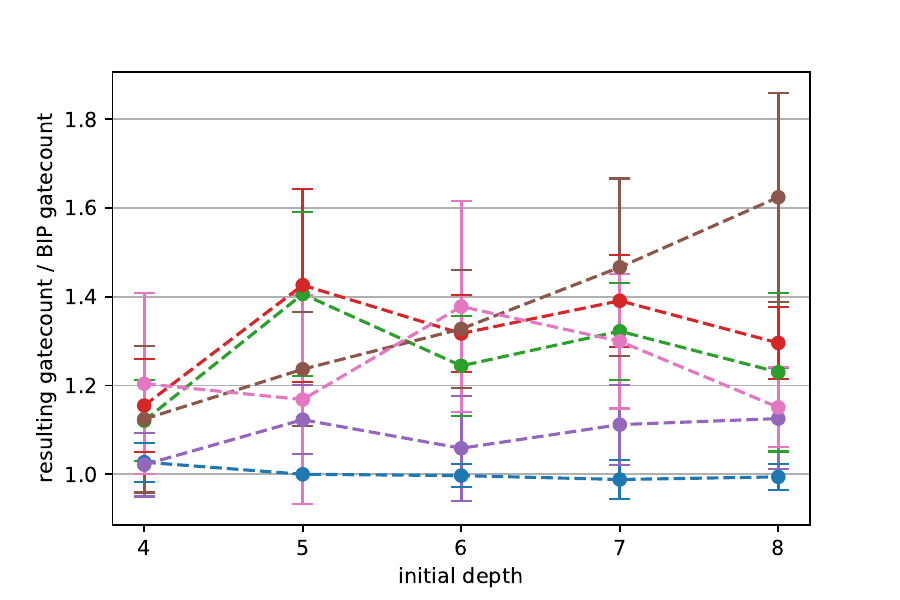}\label{fig:qvline_g}}
	\subfloat[Gatecount comparison on ring.]{\includegraphics[width=0.49\linewidth]{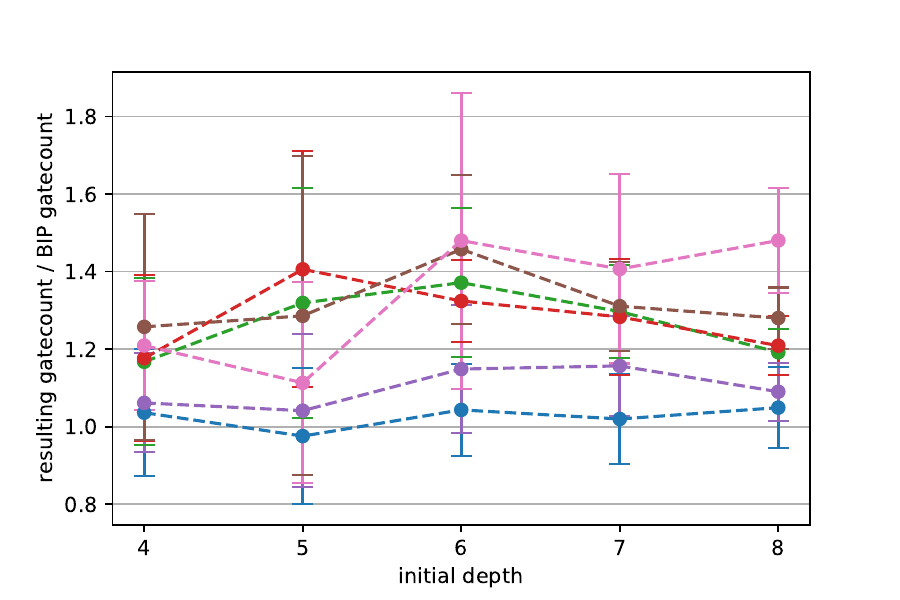}\label{fig:qvring_g}}\\
	\subfloat[Depth comparison on line.]{\includegraphics[width=0.49\linewidth]{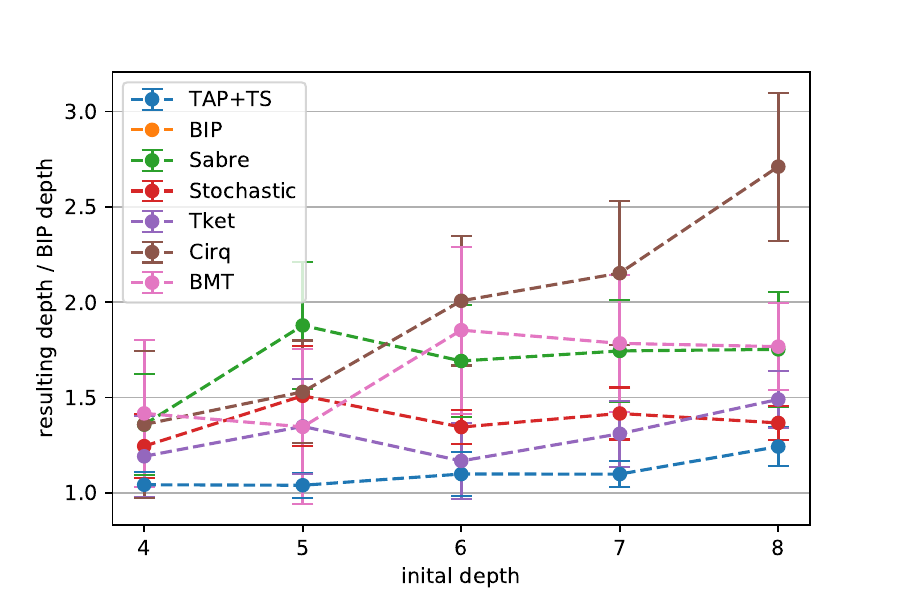}\label{fig:qvline_d}}
	\subfloat[Depth comparison on ring.]{\includegraphics[width=0.49\linewidth]{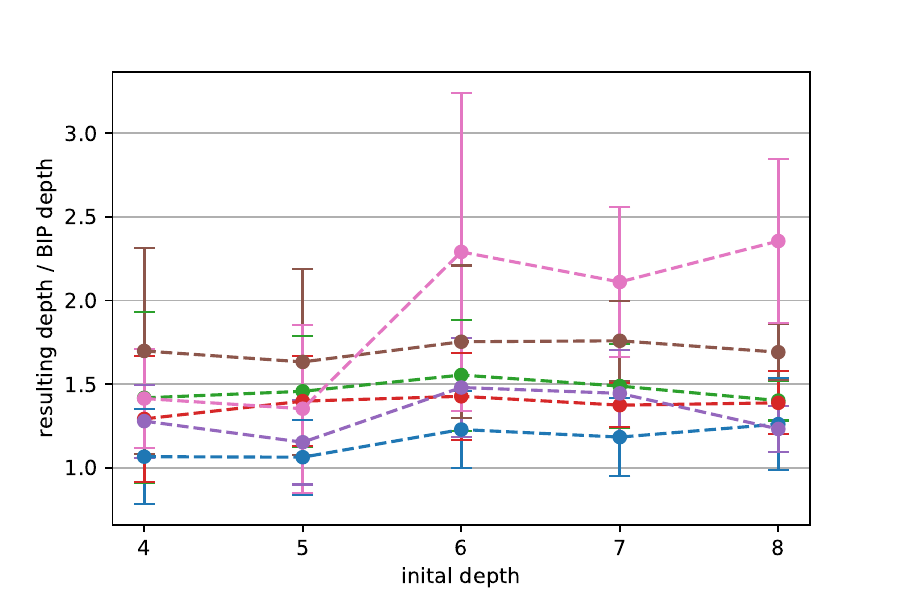}\label{fig:qvring_d}}\\
	\subfloat[Runtime comparison on line.]{\includegraphics[width=0.49\linewidth]{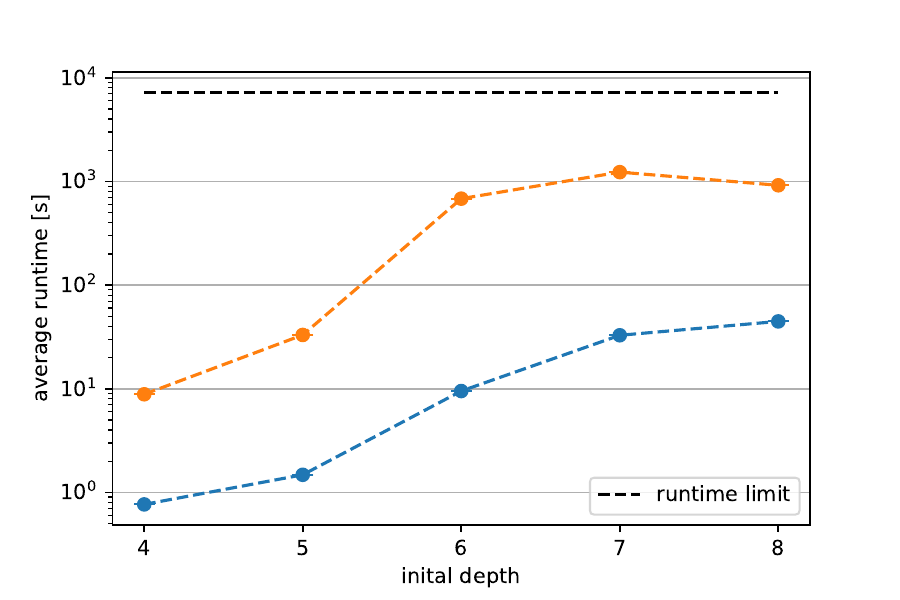}\label{fig:qvline_t}}
	\subfloat[Runtime comparison on ring.]{\includegraphics[width=0.49\linewidth]{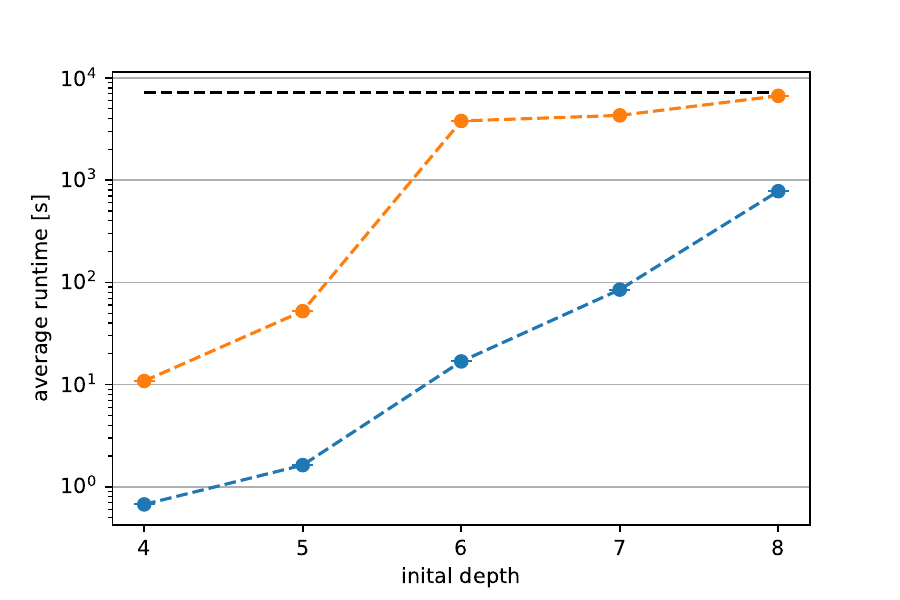}\label{fig:qvring_t}}\\
	\caption{\textcolor{blue}{Comparison of different routing algorithms on the QV benchmark set.
			Instances are QV circuits of equal width and depth (\cite{Cross2019,Moll2018})
			on a 8-vertex line (left plots) and a 8-vertex ring (right plots).
			Data points are averaged over 10~instances.
			Error bars indicate the empirical standard deviation.
			The legend in (c) holds for all plots.
			(a), (b): Gate count relative to BIP.
			(c), (d): Depth relative to BIP.
			(e), (f): Average runtime.
			The dashed line indicates the runtime limit of 7200~s.
			Other heuristics are not shown since their runtimes are in the order of 1~s.}
	}
	\label{fig:resrouting_qv}
\end{figure}
\begin{figure}
	\centering
	\subfloat[Gatecount comparison on Y graph.]{\includegraphics[width=0.49\linewidth]{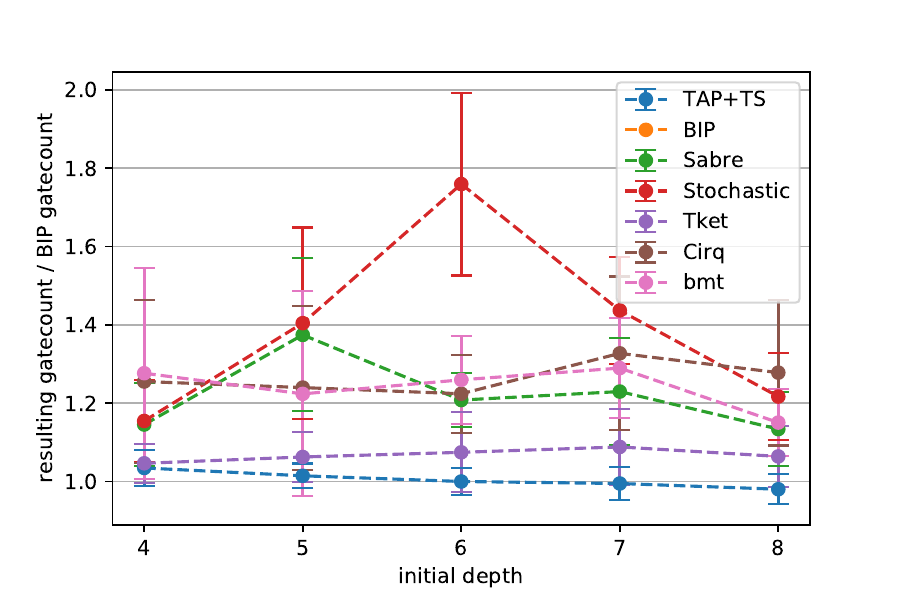}\label{fig:qvy_g}}
	\subfloat[Gatecount comparison on ladder.]{\includegraphics[width=0.49\linewidth]{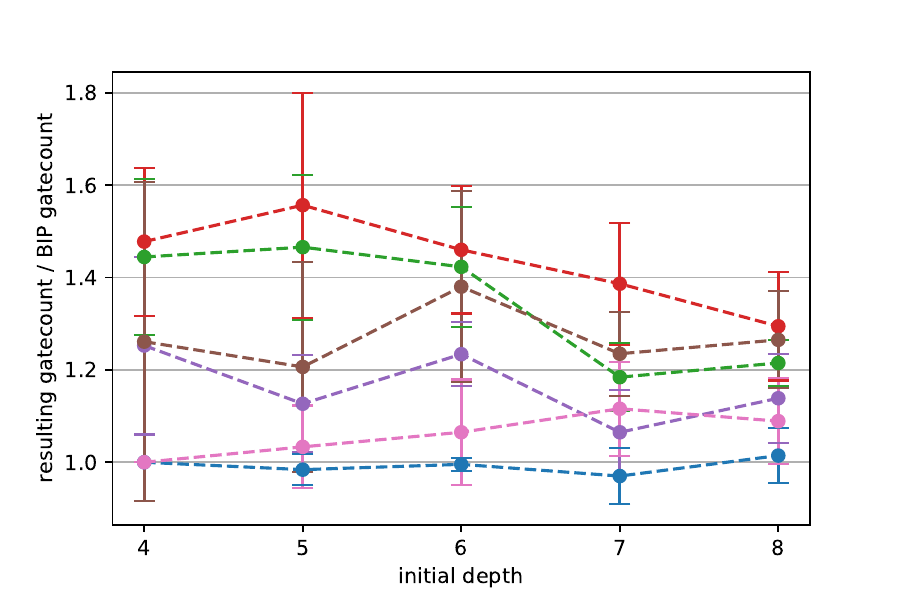}\label{fig:qvl_g}}\\
	\subfloat[Depth comparison on Y graph.]{\includegraphics[width=0.49\linewidth]{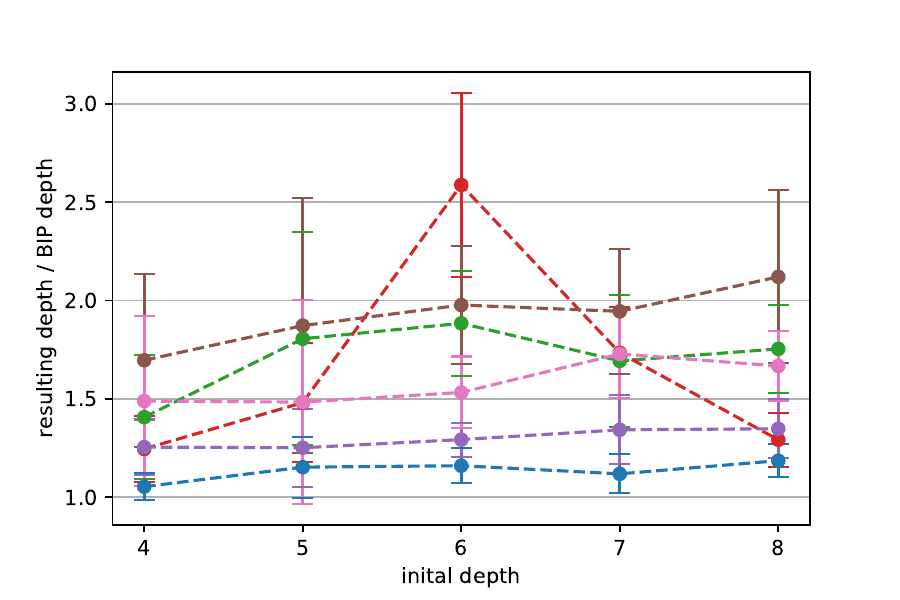}\label{fig:qvy_d}}
	\subfloat[Depth comparison on ladder.]{\includegraphics[width=0.49\linewidth]{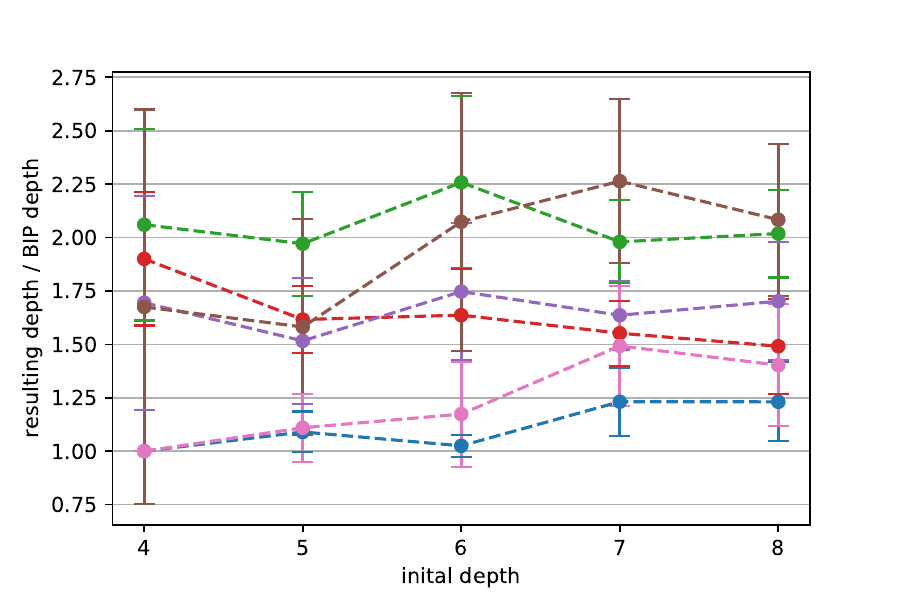}\label{fig:qvl_d}}\\
	\subfloat[Runtime comparison on Y graph.]{\includegraphics[width=0.49\linewidth]{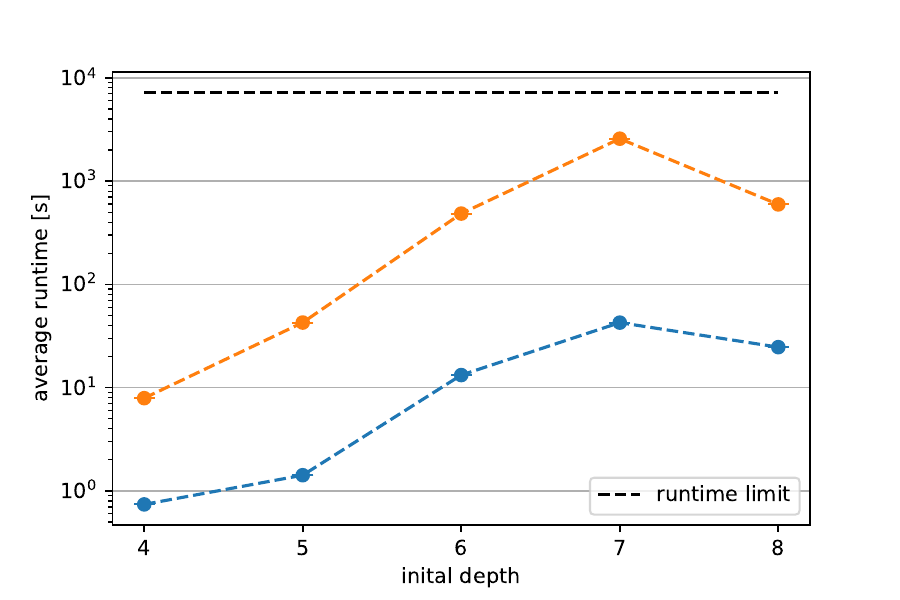}\label{fig:qvy_t}}
	\subfloat[Runtime comparison on ladder.]{\includegraphics[width=0.49\linewidth]{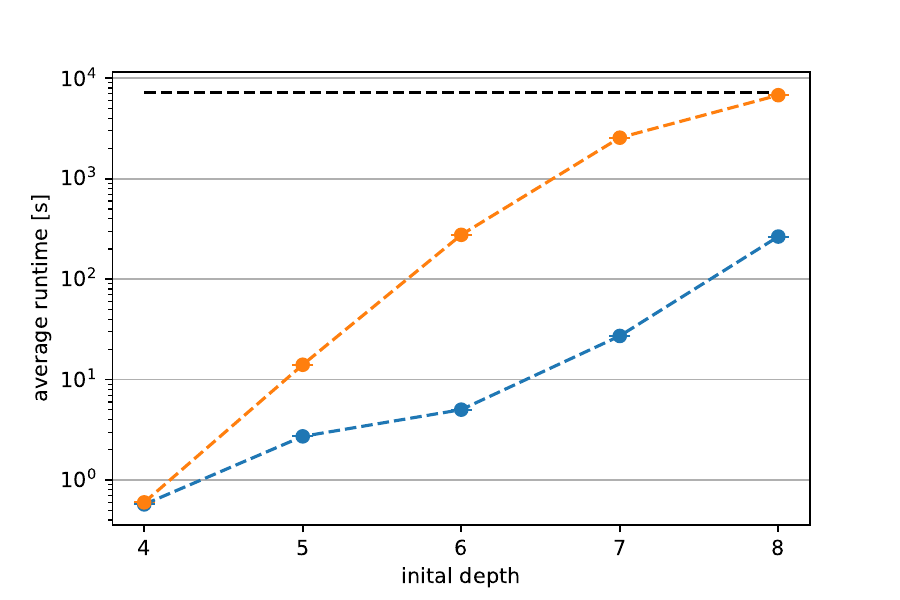}\label{fig:qvl_t}}\\
	\caption{\textcolor{blue}{Comparison of different routing algorithms on the QV benchmark set.
			Instances are QV circuits of equal width and depth (\cite{Cross2019,Moll2018})
			on a 8-vertex Y-graph (left plots) and a 8-vertex ladder-graph (right plots).
			Data points are averaged over 10~instances.
			Error bars indicate the empirical standard deviation.
			The legend in (a) holds for all plots.
			(a), (b): Gate count relative to BIP.
			(c), (d): Depth relative to BIP.
			(e), (f): Average runtime.
			The dashed line indicates the runtime limit of 7200~s.
			Other heuristics are not shown since their runtimes are in the order of 1~s.}
	}
	\label{fig:resrouting_qv2}
\end{figure}
\subsection{Performance on Real Hardware}

Lastly, we demonstrate the influence of routing quality
on the performance of current quantum hardware.
For this purpose, we compile example quantum algorithms
to a real hardware device using different routing methods
and compare the results after execution.

The example algorithms are the Quantum approximate optimization algorithm (QAOA)
with depth parameter $ p = 1 $ applied to MaxCut
on ring graphs with 7, 8 and 9 vertices.
QAOA is a well-known hybrid quantum-classical algorithm.
A detailed description is given by Farhi et al.\ in \cite{farhi2014}.
We choose QAOA for three reasons:
First, it is often considered as a promising candidate
for near-term quantum applications
because of its universality and adaptive complexity.
Second, routing is an essential requirement
when running QAOA on real hardware, since QAOA employs two-qubit gates
along the edges of an underlying problem graph
which usually is not isomorphic to the hardware connectivity graph.
This is also the case in our experiment.
Third, the approximation ratio of QAOA
is an easily accessible, single-metric performance measure.
The approximation ratio is defined as $ r = \mathbbm{E}(c) / c^* $.
Here $ \mathbbm{E}(c) $ denotes the expectation value
of the cut size~$c$ produced by QAOA,
and $ c^* $ is the size of an optimal cut.
QAOA is a parameterized quantum algorithm.
For QAOA with depth $p=1$ applied to MaxCut on ring graphs,
the optimal parameters can be calculated analytically,
\cf \cite{wang2017}.

The routing algorithms tested are the proposed method TAP+TS,
Tket and Stochastic swap.
For circuit construction, compilation and backend communication,
we use IBM's SDK Qiskit (\cite{Qiskit}).
The quantum backend in our experiment is \emph{ibmq\_ehnigen} (\cite{ibmq}), see Fig.~\ref{fig:ehnh}.
We compile each of the QAOA circuits with all three routing methods
while other compiler options remain unchanged.

Table~\ref{tab:QAOAcomp} compares the gate counts and depths
of the compiled circuits as well as the achieved approximation ratios
after execution on \emph{ibmq\_ehnigen}.
The proposed routing procedure TAP+TS performs best
both in terms of swaps added and depth increase on all instances,
followed by Tket.
For the instance with $ N = 8 $,
we note large reductions in gate count of 65\% and 69\%
when compared to Tket and Stochastic swap, respectively.
Similarly, the circuit depth reductions amount to 60\% and 68\%, respectively.
We observe that small gate count and depth correlate
with an increase of approximation ratio on all instances.
For the $ N = 8$ instance,
we observe an increase of 23\% compared to TKet. 
\textcolor{blue}{
  The influence of depth and gate count on noise
in quantum computers are currently investigated in depth by different
groups. Although the relations are not yet fully understood,
this experiment shows a clear correlation
between routing quality and quantum algorithm performance.
}

Of course, this experiment is by no means a comprehensive benchmark
of our routing method. However, it stresses the need for good compiling methods
in quantum computation on currently available hardware
and shows clear benefits of our approach.
\begin{table}
	\centering
		\begin{tabular}{ccccccccccc}
		\toprule
		$N$ & $\#_{2q}^{\text{TAP+TS}}$ & $\#_{2q}^{\text{Tket}}$& $\#_{2q}^{\text{Sto}}$ & $d^{\text{TAP+TS}}$ & $d^{\text{Tket}}$& $d^{\text{Sto}}$ &$r_{\text{ideal}}$ & $r_{\text{TAP+TS}}$ & $r_{\text{Tket}}$ & $r_{\text{St}}$\\
		\midrule
		7 & 21 & 31 & 54&22 & 28 &41 & 0.875 & 0.77 & 0.61 & 0.56  \\
		
		8 & 25 & 71 & 80 & 21 & 53 & 65 & 0.75 & 0.64 & 0.52 & 0.51  \\
		
		9 & 30 & 54 & 63 & 21 & 39 & 45 & 0.84 & 0.65& 0.58 & 0.55  \\

		\bottomrule
	\end{tabular}
\caption{Results for compiling quantum algorithms to real hardware
	with different routing algorithms.
	Compiled quantum circuits are QAOA for MaxCut on ring graphs
	with $N$ vertices.
	The target quantum device is \emph{ibmq\_ehningen}.
	Given are the total two-qubit gate counts $ \#_{2q}^{\text{TAP+TS}} $,
	$ \#^{\text{Tket}}_{2q} $, $ \#_{2q}^{\text{Sto}} $
	for the circuit routed with TAP+TS, Tket and Stochastic swap, respectively,
	as well as the corresponding depths $ d^{\text{TAP+TS}} $, $ d^{\text{Tket}} $,
	$ d^{\text{Sto}} $.
	Furthermore, we give the theoretical approximation ratio~$ r_{\text{ideal}} $
	as well as the actually achieved approximation ratios $ r_{\text{TAP+TS}} $,
	$ r_{\text{Tket}} $, $ r_{\text{St}} $.}
	\label{tab:QAOAcomp}
\end{table}

	\section{Conclusion}\label{sec:conc}
In this article, we have studied the problem of qubit routing by swap insertion.
This problem needs to be solved in the context of quantum compilation
but is practically intractable even for moderate instance sizes
when solved by an exact approach.
We have proposed a new decomposition approach
that exploits the capabilities of exact integer programming
while reducing complexity significantly,
at the expense of possibly returning suboptimal solutions.
Improvement in performance is achieved
by solving only a subproblem by integer programming.
The objective of this subproblem, called token allocation problem (TAP),
is a lower bound on the overall routing problem.
We strengthened the model of the TAP by deriving valid inequalities.
The solution of the TAP defines the second subproblem,
which is the well-known problem of token swapping.
Its solution yields the actual objective value of the overall routing problem.
For the token swapping problem,
we have enhanced an existing approximation algorithm
and have developed an exact branch-and-bound algorithm.
Combining the solutions of both subproblems
produces a solution to the overall routing problem.

Finally, we have given experimental results
showing the applicability of the proposed algorithm.
Although it being a heuristic approach,
we observed close-to-optimal solutions outperforming other heuristics.
This improvement comes at the expense of increased solution time;
however, the solution time is still much smaller than for exact approaches.
Moreover, the lower bounds computed in the TAP turn out to be tight,
which we consider as a reason for the good performance
of the proposed decomposition method.

Our experiments on actual quantum hardware show the impact
of high-quality solutions to qubit routing for quantum computation
on current hardware.
For the example algorithm QAOA, we observed an increase in approximation ratio
when compiled with the proposed routing method compared to standard heuristics.

A possible field of future research is the structural analysis of the TAP
in order to further strengthen its integer programming formulation.
Additionally, practical experience with current quantum hardware
suggests more complicated cost functions to be considered:
specific qubits and gates suffer from higher error rates than others.
Even more complicated, executing gates on specific qubit pairs simultaneously
results in high error rates -- a phenomenon known as \emph{crosstalk}.
Incorporating these effects in the routing algorithm
is an additional direction of future research.

	\section*{Acknowledgements}
	This research has been supported
	by the Bavarian Ministry of Economic Affairs, Regional Development
	and Energy with funds from the Hightech Agenda Bayern
	as well as the Center for Analytics -- Data -- Applications
	(ADA-Center) within the framework
	of \qm{BAYERN DIGITAL~II} (20-3410-2-9-8).
	The research is part of the Munich Quantum Valley, which is supported by the Bavarian state government with funds from the Hightech Agenda Bayern Plus.
	\bibliographystyle{abbrvurl}
	\bibliography{bib_compiling}
	
\end{document}